\newif\ifisfullpaper
\isfullpaperfalse

\ifisfullpaper
\documentclass[pdflatex,sn-mathphys-num]{sn-jnl}
\else 
\documentclass[pdflatex,sn-mathphys-num]{sn-jnl}
\newgeometry{margin=1in}
\fi

\usepackage{tabularx}

\usepackage{graphicx}%
\usepackage{multirow}%
\usepackage{amsmath,amssymb,amsfonts}%
\usepackage{amsthm}%
\usepackage{mathrsfs}%
\usepackage[title]{appendix}%
\usepackage{xcolor}%
\usepackage{textcomp}%
\usepackage{manyfoot}%
\usepackage{booktabs}%

\usepackage{xcolor}
\definecolor{teal}{HTML}{008080}
\usepackage{graphicx}
\usepackage{amsmath}
\usepackage{amsthm}
\usepackage{multirow}
\usepackage{caption}
\usepackage{subcaption}

\allowdisplaybreaks

\usepackage{hyperref}
\hypersetup{
   colorlinks=true,
   citecolor=teal   
}

\makeatletter
\newcommand{\dotdiv}{\mathbin{\text{\@dotdiv}}}

\newcommand{\@dotdiv}{%
  \ooalign{\hidewidth\raise1ex\hbox{.}\hidewidth\cr$\m@th-$\cr}%
}
\makeatother

\usepackage[textsize=scriptsize, textwidth=2.1cm]{todonotes}
\setuptodonotes{color=green!40}

\usepackage{tikz}
\usepackage{tkz-euclide}
\usetikzlibrary{chains,positioning}

\newcommand{\descr}[1]{\vspace{0.2cm} \noindent \textbf{#1}}

\usepackage[noend, ruled]{algorithm2e} % for algorithms

\newtheorem{proposition}{Proposition}
\newtheorem{theorem}{Theorem}

\theoremstyle{definition}
\newtheorem{example}{Example}
\newtheorem{definition}{Definition}

\newcommand{\vtr}[1]{\mathbf{#1}}
\newcommand{\norm}[1]{\lVert{#1}\rVert}

\newcommand{\database}{\mathcal{D}}

\begin{document}
%
%\title{Committed Hamming Property Preserving Hashing and Its Applications\thanks{Supported by Macquarie University Cyber Security Hub.}}
\title{Property-Preserving Hashing for \texorpdfstring{$\ell_1$}{l1}-Distance Predicates: Applications to Countering Adversarial Input Attacks}
%\author{First Author\inst{1}\orcidID{0000-0001-9649-943X} \and
%Second Author\inst{2,3}\orcidID{1111-2222-3333-4444} \and
%Third Author\inst{3}\orcidID{2222--3333-4444-5555}}
%
%\institute{Princeton University, Princeton NJ 08544, USA \and
%Springer Heidelberg, Tiergartenstr. 17, 69121 Heidelberg, Germany

\author*[1]{\fnm{Hassan} \sur{Asghar}}\email{hassan.asghar@mq.edu.au}
\author[1]{\fnm{Chenhan} \sur{Zhang}}\email{chzhang@ieee.org}
\author[1]{\fnm{Dali} \sur{Kaafar}}\email{dali.kaafar@mq.edu.au}

\affil*[1]{\orgdiv{School of Computing}, \orgname{Macquarie University}, \orgaddress{\country{Australia}}}
% \affil[2]{\orgdiv{School of Computing}, \orgname{Macquarie University}, \orgaddress{\country{Australia}}}
% \affil[3]{\orgdiv{School of Computing}, \orgname{Macquarie University}, \orgaddress{\country{Australia}}}

% \author{Hassan Jameel Asghar}
% \email{hassan.asghar@mq.edu.au}
% \affiliation{%
% 	\institution{Macquarie University}
% 	\city{Sydney}
% 	\country{Australia}
% }
% \author{Chenhan Zhang}
% \email{chzhang@ieee.org}
% \affiliation{%
% 	\institution{Macquarie University}
% 	\city{Sydney}
% 	\country{Australia}
% }
% \author{Dali Kaafar}
% \email{dali.kaafar@mq.edu.au}
% \affiliation{%
% 	\institution{Macquarie University}
% 	\city{Sydney}
% 	\country{Australia}
% }

\abstract{Perceptual hashing is widely used to detect whether an input image is similar to a reference image with a variety of security applications. Recently, it has been shown to succumb to adversarial input attacks which make small imperceptible changes to the input image yet the hashing algorithm does not detect its similarity to the original image. Property-preserving hashing (PPH) is a recent construct in cryptography, which preserves some property (predicate) of its inputs in the hash domain. Researchers have so far shown constructions of PPH for Hamming distance predicates, e.g., the predicate which outputs $1$ if two inputs are within Hamming distance $t$. A key feature of PPH is its strong correctness guarantee, i.e., the probability that the predicate will not be correctly evaluated in the hash domain is negligible. Motivated by the use case of detecting similar images under adversarial setting, we propose the first PPH construction for an $\ell_1$-distance predicate. Roughly, this predicate checks if the two one-sided $\ell_1$-distances between two images are within a threshold $t$. Since many adversarial attacks use $\ell_2$-distance (related to $\ell_1$-distance) as the objective function to perturb the input image, by appropriately choosing the threshold $t$, we can force the attacker to add considerable noise to evade detection, and hence significantly deteriorate the image quality. Our proposed scheme is highly efficient, and runs in time $\mathcal{O}(t^2)$. For grayscale images of size $28 \times 28$, we can evaluate the predicate in $0.0784$ seconds when pixel values are perturbed by up to $1 \%$. For larger RGB images of size $224 \times 224$, by dividing the image into $1,000$ blocks, we achieve times of $0.0128$ seconds per block for $1 \%$ change, and up to $0.2641$ seconds per block for $14\%$ change. Furthermore, the time to process the entire image can be considerably improved since the scheme is highly parallel.}    
\keywords{Property-preserving hashing, adversarial attacks, error-correcting codes}

\maketitle

% ---- Introduction ----
%
\section{Introduction}
\label{sec:Intro}
\ifisfullpaper
Consider
\else
Consider\footnote{This is the preprint of the paper with the same title, which has been accepted for publication in Cryptography and Communications from Springer Nature.}
\fi
a scenario in which an image needs to be checked against a database of images for any similarity. Privacy demands that both the image and the database not be revealed during this process. This scenario stems from several real-world use cases. For instance, this is required in face recognition for border control where the identity of a passenger is verified against a gallery of photos from other passengers on the same flight.\footnote{See ``2024 Update on DHS’s Use of Face Recognition \& Face Capture Technologies'' at \url{https://www.dhs.gov/archive/news/2025/01/16/2024-update-dhss-use-face-recognition-face-capture-technologies}.} Likewise, a cloud service provider may wish to ensure that an image uploaded to its cloud is not one of the flagged images in its database~\cite{struppek2022neuralhashbreak}. A possible solution to this is via \emph{perceptual hashing}~\cite{klinger2013phash, hao2021pph-attack, struppek2022neuralhashbreak}, variants of which have been used by Microsoft,\footnote{See PhotoDNA at \url{https://www.microsoft.com/en-us/photodna}.} Meta,\footnote{See \url{https://about.fb.com/news/2019/08/open-source-photo-video-matching/}.} and Apple.\footnote{See ``CSAM Detection -- Technical Summary'', at \url{https://www.apple.com/child-safety/pdf/CSAM_Detection_Technical_Summary.pdf}.} Perceptual hashing is a type of \emph{locality sensitive hashing} (LSH)~\cite{indyk1998ann} which produces similar hashes to perceptually similar images, unlike cryptographic hash functions. The result is that we can not only check images for their similarities, but can also do so more efficiently using their succinct hash digests. 
%\has{How about the hiding property? Might not be that big a concern. See~\cite[\S 6]{struppek2022neuralhashbreak}}.

Recently, several attacks have been demonstrated on perceptual hashing. One type of attack, called an \emph{evasion attack}, 
slightly perturbs an image to construct a perceptually similar image but with dissimilar hash digests under the perceptual hashing scheme~\cite{struppek2022neuralhashbreak, hao2021pph-attack}. Many perceptual hash functions involve two main steps (among others): extracting features from the image and creating a hash of the feature vector~\cite{hao2021pph-attack}. The guarantee that two perceptually similar images produce a similar hash digest is only in probability, which in practice is not negligible. Thus, there is significant space available to the attacker for image alterations to launch an evasion attack. For instance, the basic attack in~\cite{hao2021pph-attack} formulates the competing requirements of perceptual similarity and dissimilarity of the hash digests as an optimization problem, the solution to which is the required adversarial image. 

As mentioned above, the reason for the success of evasion attacks is that perceptual hashing, and locality sensitive hashing in general, does not tend to have negligible \emph{correctness error}. Informally, let $\vtr{x}$ and $\vtr{y}$ be two perceptually similar images, and let $h$ be a perceptual hash function, then ideally the requirement is that $\Pr[h(\vtr{x}) = h(\vtr{y}] \geq 1 - \epsilon$~\cite{zauner2010ph-thesis, hao2021pph-attack, holmgren2022nearly}. However, there are strong lower bounds suggesting that $\epsilon$ cannot be made negligibly small~\cite{o2014lowerbound-ph, holmgren2022nearly}. A related notion of hashing is \emph{property-preserving hashing} (PPH)~\cite{boyle2018adversarially}. Such hash functions have the property that the hash digests of two inputs preserve some predicate of the two inputs. For instance, a predicate that outputs 1 if the Hamming distance of its two inputs are within a certain threshold. Moreover, the definition postulates that the correctness error be negligible. Generally, these hash functions are sampled from a larger PPH family of functions. Recent works have constructed robust versions of PPH in which the adversary can choose inputs depending on the description of the hash function chosen from the family~\cite{boyle2018adversarially, holmgren2022nearly, fleischhacker2021robust, fleischhacker2022property}. These works have focused on Hamming distance predicates. 

Our interest in PPH for the aforementioned application of checking similarity of images stems from the fact that adversarial attacks to induce image misclassification introduce small perturbations in the images by using a distance metric such as the $\ell_2$-distance in the objective function~\cite{madry2018towards, carlini2017l2}. Furthermore, researchers in this space have also used the $\ell_2$-norm as a measure of perceptual similarity between images~\cite{hao2021pph-attack, carlini2017l2}. Thus, if we have a PPH function family that preserves Euclidean distance predicates we can be assured that any evasion attacks possible in the hash space of the images are precisely those that are possible in the original space over the Euclidean distance. As a result, setting an appropriate threshold for similarity over the Euclidean distance ensures that the attacker can only succeed by substantially distorting the adversarial image, i.e., reduced perceptual similarity to the original image. Our contributions are as follows

\begin{itemize}
    \item We propose the first PPH family for the asymmetric $\ell_1$-distance predicate for $n$-element vectors with each element from the set $\{0, 1, \ldots, q-1\}$. Roughly, given two vectors $\vtr{x}$ and $\vtr{y}$, the asymmetric $\ell_1$-distance predicate outputs 1 if both $\norm{\vtr{x} \dotdiv \vtr{y}}_1$ and $\norm{\vtr{y} \dotdiv \vtr{x}}_1$ are less than $t/2$ for some threshold $t$, where $\vtr{x} \dotdiv \vtr{y}$ is the vector whose $i$th element is defined as $\max\{x_i - y_i, 0\}$. Prior to this work, the only known PPH constructions are for Hamming distance predicates~\cite{boyle2018adversarially, holmgren2022nearly, fleischhacker2021robust, fleischhacker2022property}. Our scheme is based on $\ell_1$-error correcting codes from Tallini and Rose~\cite{tallinil1codes}. The $\ell_1$-distance metric is related to $\ell_2$-distance (see Section~\ref{sec:pre}), and can be used as a proxy for adversarial attacks or perceptual similarity. 
    \item We show that our PPH construction is robust to one-sided errors: the predicate outputs 1 yet the PPH evaluates to 0. This property is important to prevent the aforementioned evasion attacks. We also give evidence that the other one-sided error can be made practically small in the non-robust setting. This error is important to rule out collisions, i.e., two dissimilar images that produce the same hash. We discuss collision attacks and inverting the hash function, i.e., finding the input vector given the hash digest, and show some evidence that they may be computationally expensive.  
    \item We prove lower bounds on the possible compression, i.e., the length of the digest, to preserve any $\ell_1$-distance predicate. Our scheme produces hash digests of length $t \log_2 n$, which is considerably less than $n \log_2 q$ (size of images) if $t$ is small. We show that for practical parameters this is close to the lower bounds for small $t$. For larger $t$ the digest size is large, but on par with Hamming distance PPHs which achieve compression of around $t \log_2 n$~\cite{holmgren2022nearly}.\footnote{We remark that the digest size cannot be  independent of the size of the image $n$ and the threshold $t$, as the digest needs to contain information about the distance. However, in many applications images are resized to a fixed resolution before being processed, and hence we can assume that all images are of the same size in our application.}
    \item We implement our scheme using the Python library \verb+galois+~\cite{Hostetter_Galois_2020} and show the computational time of our scheme against increasing values of $t$. We further implement two adversarial evasion attacks and two generic image transformations over the public Imagenette dataset~\cite{imagenette}, and show how the scheme can prevent such attacks with a tradeoff between adversarial attack prevention and computational time.
\end{itemize}

\section{Preliminaries}
\label{sec:pre}
\descr{Images.} 
%Let $q, m$ be positive integers. A pixel can be modelled as an element of the set $\mathbb{Z}_q^m$, where $\mathbb{Z}_q = \{0, 1, \ldots, q - 1\}$. An image $\mathbf{x}$ is defined as a member of the product space $(\mathbb{Z}_q^m)^n$. We can \emph{flatten} this space by appending each pixel after the other. Hence for our purposes 
Let $q, n$ be positive integers. The image space is the set $\mathbb{Z}_q^n$, where $\mathbb{Z}_q = \{0, 1, \ldots, q - 1\}$, and $n$ is the product of the number of elements in a pixel times the total number of pixels in the image. While we assume that $q \geq 2$, we are not interested in $q = 2$, as binary images can be handled via Hamming distance predicates. 

\begin{example}
\label{ex:images}
Consider a set $S$ of images of size $28 \times 28$ with each pixel having a binary value, white or black. Each pixel is therefore from $\mathbb{Z}_2$. The $28 \times 28$ pixel matrix can be ``flattened'' to a vector of 784 pixels. Thus each image in $S$ is in $\mathbb{Z}_2^{784}$. Consider instead that $S$ contains only grayscale images. Each pixel now has a grayscale value in $\{0, 1, \ldots, 255\}$. We can represent images in $S$ as vectors from $\mathbb{Z}_{256}^{784}$ (with $q = 256$). Consider now that the images in $S$ are RGB colored images. Each pixel is now a vector containing R, G and B values, each in $\{0, 1, \ldots, 255\}$, i.e., from the vector space $\mathbb{Z}_{256}^3$. Thus the images in $S$ are vectors from the product space $(\mathbb{Z}_{256}^{3})^{784}$. This can be flattened further to obtain the product space $\mathbb{Z}_{256}^{2352}$.
\qed
\end{example}

\descr{Metrics.} The Hamming, $\ell_1$ and $\ell_2$-distances between vectors $\vtr{x}, \vtr{y} \in \mathbb{Z}_q^n$ are denoted by $\norm{\vtr{x} - \vtr{y}}_0$, $\norm{\vtr{x} - \vtr{y}}_1$ and $\norm{\vtr{x} - \vtr{y}}_2$, respectively. A refresher on these is given in Appendix~\ref{app:useful}.

\descr{Dot Product.} Let $\vtr{x}, \vtr{y} \in \mathbb{Z}_q^n$, then their dot product is defined as $\langle \vtr{x}, \vtr{y} \rangle = \sum_{i = 1}^n x_i y_i$. 
The Cauchy-Schwarz inequality states that 
\[
\left| \langle \vtr{x}, \vtr{y} \rangle \right| \leq \norm{\vtr{x}}_2 \norm{\vtr{y}}_2. 
\]
See for example~\cite[\S 2.2]{shurman2016calculus}.

\descr{Relation Between Norms.} The following result relates the $\ell_1$-norm to the $\ell_2$-norm.
\begin{proposition}
\label{prop:l1-l2-relation}
Let $\vtr{x}$ and $\vtr{y}$ be images. Then,
\begin{equation}
\label{eq:l1-l2-relation}
 \norm{\vtr{x} - \vtr{y}}_2 \leq \norm{\vtr{x} - \vtr{y}}_1 \leq \sqrt{n}\norm{\vtr{x} - \vtr{y}}_2.   
\end{equation}
Furthermore, these bounds are tight. 
\end{proposition}
\begin{proof}
Let $\vtr{z} = \vtr{x} - \vtr{y}$. Then, equivalently, we need to show that
\[
\norm{\vtr{z}}_2 \leq \norm{\vtr{z}}_1 \leq \sqrt{n}\norm{\vtr{z}}_2
\]
Note that $\vtr{z}$ may not be an image, i.e., $\vtr{z}$ may not be a member of $\mathbb{Z}_q^n$. Now we see that
\begin{align*}
  \norm{\vtr{z}}_1^2 &= \left( \sum_{i = 1}^n |z_i| \right)^2 \\
  &= \left( \sum_{i = 1}^n |z_i| \right) \left( \sum_{i = 1}^n |z_i| \right)\\
  &=  \sum_{i = 1}^n |z_i|^2  +  \sum_{i, j : i \neq j} |z_i||z_j| \\
  &\geq \sum_{i = 1}^n |z_i|^2 =  \norm{\vtr{z}}_2^2,
\end{align*}
from which it follows that $\norm{\vtr{z}}_2 \leq \norm{\vtr{z}}_1$. For the second inequality, let $\vtr{1}$ be the vector of all 1's, and let $\vtr{b}$ be such that $b_i = |z_i|$ for all $i$. Then,
\begin{equation*}
  \norm{\vtr{z}}_1 = \sum_{i=1}^n |z_i| = \langle \vtr{1}, \vtr{b} \rangle  \leq \norm{\vtr{1}}_2 \norm{\vtr{b}}_2 = \sqrt{n} \norm{\vtr{z}}_2,
\end{equation*}
where we have used the Cauchy-Schwarz inequality. To show that the bounds are tight, let us first consider the inequality on the left. Assume that for some positive constant $c$ we have $c \norm{\vtr{x} - \vtr{y}}_2 \leq \norm{\vtr{x} - \vtr{y}}_1$. Consider the two images $\vtr{x} = (1, 0, \ldots, 0)$ and $\vtr{y} = \vtr{0}$, where $\vtr{0}$ is the zero vector. Then,
\[
c \norm{\vtr{x} - \vtr{y}}_2 \leq \norm{\vtr{x} - \vtr{y}}_1 \Rightarrow c \sqrt{1} \leq 1 \Rightarrow c \leq 1.
\]
Thus, $c = 1$ is tight. For the RHS, assume that for some positive constant $c$, we have $ \norm{\vtr{x} - \vtr{y}}_1 \leq c \norm{\vtr{x} - \vtr{y}}_2$.  Consider the two images $\vtr{x} = (q-1, q-1, \ldots, q-1)$ and $\vtr{y} = \vtr{0}$. Then,
\[
\norm{\vtr{x} - \vtr{y}}_1 \leq c \norm{\vtr{x} - \vtr{y}}_2 \Rightarrow n(q-1) \leq c \sqrt{n}(q-1) \Rightarrow c \geq \sqrt{n}. 
\]
\end{proof}

Thus, it suffices to focus on the $\ell_1$-norm. For instance, if we want to preserve $\norm{\vtr{x} - \vtr{y}}_2 \leq t'$ for some threshold $t'$ , we can set $t = \sqrt{n}t'$ as the threshold for the $\ell_1$-norm. 

%, as any PPH that works for the $\ell_1$-norm can be . We will focus on the Eucilidean norm. 

%\descr{Displacement Vectors.} In the adversarial attacks considered, we will add a displacement vector $\boldsymbol{\delta} \in \mathbb{R}^n$ to images $\vtr{x}$. We define this as follows. For $x \in \mathbb{Z}_q$ and $\delta \in \mathbb{R}$, define $\lfloor x + \delta \rceil_q$ as the nearest integer to $x + \delta$ in $\mathbb{Z}_q$. Then $\vtr{x} + \boldsymbol{\delta}$ is an image whose $i$th element is $\lfloor x_i + \delta_i \rceil_q$. 

\descr{Min-Entropy.} To prove lower bounds on the amount of compression achievable under the Hamming distance, Holmgren et al~\cite{holmgren2022nearly} use the notion of average min-entropy as defined in~\cite{dodis-min-entropy}. The min-entropy $H_\infty(X)$ of the random variable $X$ is defined as $-\log_2 (\max_x \Pr(X = x))$. For a pair of random variables $X$ and $Y$, the average min-entropy of $X$ given $Y$ is defined as $H_\infty(X \mid Y) = -\log_2 (\sum_{y \in Y} \Pr(Y = y) \cdot (\max_x \Pr(X= x \mid Y= y))$. The following result is from~\cite{holmgren2022nearly, dodis-min-entropy}.

\begin{proposition}[Dodis et al~\cite{dodis-min-entropy}]
\label{prop:dodis}
Let $X, Y, Z$ be random variables with $Z$ having support over a binary string of length $m$. Then
\[
H_\infty(X \mid Y, Z) \geq H_\infty(X \mid Y) - m
\]
\end{proposition}

\subsection{Scenario and Threat Model}
\label{subsec:require}
We assume a database $\database$ of $N$ images $\vtr{x}_1, \ldots, \vtr{x}_{N}$ hosted by a server. 
%We consider there to me a perceptual similarity function $\mathsf{sim}: \mathbb{Z}_q^n \times \mathbb{Z}_q^n \rightarrow \{0, 1\}$, which outputs 1 if two images are similar and 0 otherwise. This function simply models our ground truth and hence we shall only be using it as an interface without making it precise. We assume that $\mathsf{sim}(\vtr{x}_i, \vtr{x}_j) = 0$ for all $i, j \in [N], i \neq j$. That is, all the images in the database are dissimilar. 
A predicate $P : \mathbb{Z}_q^n \times \mathbb{Z}_q^n \rightarrow \{0, 1\}$ is a function applied to a pair of images. For instance, one such predicate for two images $\vtr{x}, \vtr{y}$ is:
\begin{equation}
\label{eq:l1-distance-pred}
P(\vtr{x}, \vtr{y}) = \begin{cases}
    1, &\text{ if } \norm{\vtr{x} - \vtr{y}}_1 \leq t, \\
    0, & \text{ otherwise}
\end{cases}
\end{equation}
%The parameter $t$ is omitted in the description of $P$, but shall be clear from the context. 
The predicate used in our scheme is different from the one above, but still based on the $\ell_1$-distance, as we shall see later.
%We shall use a threshold $t$ such that the \emph{empirical error} of the predicate $P$ on the database $\database$ is zero, or very close to 0, defined as
% \begin{equation}
% \label{eq:emp-error-p}
% \text{err}_P(\database) = \sum_{1 \leq i < j \leq N} \frac{P(\vtr{x}_i, \vtr{x}_j)}{\binom{N}{2}}
% \end{equation}
Given any input image $\vtr{y}$ from a client, the server checks if it satisfies the predicate $P$ against any image $\vtr{x}_i \in \database$. 

We consider a hash function family $\mathcal{H} = \{ h : \mathbb{Z}_q^n \rightarrow \{0, 1\}^m \}$. Given $\vtr{x} \in \mathbb{Z}_q^n$, we call $h(\vtr{x})$ for some $h \in \mathcal{H}$, the \emph{hash digest} of the image $\vtr{x}$. Associated with the family, there is a deterministic polynomial time algorithm which evalutes the function $\mathsf{eval}_h : \{0, 1\}^m \times \{0, 1\}^m \rightarrow \{0, 1\}$. %Intuitively, $h$ should be a hash function for the predicate $P$, such that any computations done using $P$ could also be done on $h$. More precisely, but still 
Informally, the hash function $h$ should satisfy the following properties:

\begin{enumerate}
    \item \textit{Compression:} The output $h(\vtr{x})$ should be compressed, i.e., the size of $h(\vtr{x})$ should be less than the size of the image $\vtr{x}$.
    \item \textit{Hiding:} It should be hard to find $\vtr{x}$ from $h(\vtr{x})$.
    \item \textit{Property-preservation:} The function $\mathsf{eval}_{h}$ should be such that $\mathsf{eval}_h(h(\vtr{x}), h(\vtr{y})) = P(\vtr{x}, \vtr{y})$ with high probability.
%    \item \textit{Robustness:} For any adversary $\adv{A}$ who is given $h \in \mathcal{H}$, it is hard to come up with a pair of images $\vtr{x}$ and $\vtr{y}$ such that $\norm{\vtr{x} - \vtr{y}}_2 \leq \varepsilon$ and $\eval{h(\vtr{x}), h(\vtr{y}}) = 0$. 
\end{enumerate}
The first requirement is for utility, as one would want the hash function to at least compress the input space. The other two requirements are for privacy. We want the input image $\vtr{x}$ to be hidden from the server, with only the output of the predicate $P$ revealed. The third requirement is related to an adversarial user who wishes to submit an image $\vtr{y}$ such that there is a mismatch between $P(\vtr{x}, \vtr{y})$ for some $\vtr{x} \in \database$, and the $\mathsf{eval}_h$ function.  

\descr{Adversarial Goals.} We consider the following attacks. 
\begin{enumerate}
    \item \emph{Inversion attack:} Given the hash digest $h(\vtr{x})$ of an image $\vtr{x}$ chosen uniformly at random from $\mathbb{Z}_q^n$, find $\vtr{x}$. 
    \item \emph{Evasion attack:} Given $h \in \mathcal{H}$, find a pair of images $\vtr{x}$ and $\vtr{y}$, such that $P(\vtr{x}, \vtr{y}) = 1$ yet $\mathsf{eval}_h(h(\vtr{x}), h(\vtr{y})) = 0$.
    \item \emph{Collision attack:} Given $h \in \mathcal{H}$, find a pair of images $\vtr{x}$ and $\vtr{y}$ such that $P(\vtr{x}, \vtr{y}) = 0$ yet $\mathsf{eval}_h(h(\vtr{x}), h(\vtr{y})) = 1$.
\end{enumerate}
The first goal is in a similar flavour to inverting passwords given their hash digests. %Usually, passwords are hashed with random strings $r$ to ensure that the property is robust against an attacker who knows the possibly small universe of input values. We discuss how this can be done in our case. But for the above definition we avoid this by requiring the image $\vtr{x}$ to be sampled uniformly at random. 
The defence against the second attack follows directly from the definition of PPH. This is the main attack that our scheme seeks to prevent. Prevention of the third attack is also desirable. We show that while this may still be possible, the probability of collisions can be made extremely low in the non-robust setting, i.e., two arbitrary images deemed similar under $h$. 

% Suppose $h$ satisfies these properties, and consider the detector $\detector_{h, \varepsilon}$ defined as:

% \medskip
% \begin{algorithmic}[1]
% \State $X \leftarrow \emptyset$
% \For{$h(\vtr{x}) \in h(\database{D})$}
%     \State \textbf{if} $\eval{h(\vtr{x}), h(\vtr{y})} = 1$ \textbf{then} $X \leftarrow X \cup \{\ind{h(\vtr{x})}\}$
% \EndFor
% \State \Return $X$
% \end{algorithmic}
% \medskip

% Then, this detector is correct with high probability (mimics the performance of the ideal detector), and cannot be defeated by any adversary with high probability. Properties 3 and 4 are similar to properties of a robust property preserving hash (RPPH) function (family) [REF]. The second property is also desired in RPPHs but not an explicity goal [REF]. The first two properties are necessary for our application.  

\subsection{Property-Preserving Hashing and Compression Bounds}
\label{sub:pph-ham-compress}
We recall the definitions of PPH and robust PPH (RPPH) from~\cite{holmgren2022nearly, boyle2018adversarially}. We assume $n$ and $m$ to be polynomials in a security parameter $\lambda$. 
\begin{definition}[Property Preserving Hash (PPH)]
\label{def:pph}
A $(b, n, m)$-property preserving hash family $\mathcal{H} = \{ h : \{0, 1\}^n \rightarrow \{0, 1\}^m$ for a predicate $P : \{0, 1\}^n \times \{0, 1\}^n \rightarrow \{0, 1\}$ is a family of efficiently computable functions with the following algorithms
\begin{itemize}
    \item $\mathsf{samp}(1^\lambda)$ is a probabilistic polynomial time algorithm that outputs a random $h \in \mathcal{H}$.
    \item $\mathsf{eval}_h(y_1, y_2)$ is a deterministic polynomial time algorithm that for an $h \in \mathcal{H}$ and $y_1, y_2 \in \{0, 1\}^m$ outputs a single bit.
    \item \emph{$b$-Correctness:} For a bit $b \in \{0, 1\}$, for any $h \in \mathcal{H}$ and for all $x_1, x_2 \in \{0, 1\}^n$ we have
\[
\Pr_{h \leftarrow \mathsf{samp}(1^\lambda)} \left[ P(x_1, x_2) \neq \mathsf{eval}_h(h(x_1), h(x_2)) \mid P(x_1, x_2) = b \right] = \mathsf{negl}(\lambda)
\]
\end{itemize}  \qed  
\end{definition}
Note that we have slightly modified the definition of PPH from~\cite{holmgren2022nearly, boyle2018adversarially} to include two-sided correctness. This is because, even though our scheme is $1$-correct, $0$-correctness is only guaranteed with a small but non-negligible probability.  
A PPH is considered to be an RPPH if it further satisfies the following definition.
\begin{definition}[Robust Property Preserving Hash (RPPH)]
\label{def:rpph}
A $(b, n, m)$-PPH family is a robust $(b, n, m)$-property preserving hash family if for all probabilistic polynomial time algorithms $\mathcal{A}$
\[
\Pr_{\substack{h \leftarrow \mathsf{samp}(1^\lambda)\\x_1, x_2 \leftarrow \mathcal{A}(h)}} \left[ P(x_1, x_2) \neq \mathsf{eval}_h(h(x_1), h(x_2)) \mid P(x_1, x_2) = b \right] = \mathsf{negl}(\lambda)
\] \qed
\end{definition}
The main difference between an RPPH and a PPH is that the inputs $x_1, x_2$ that bring about a mismatch between the predicate and the evaluation function are adversarially chosen in the former who is also given the description of the sampled hash function $h$. 

\descr{Differences from Perceptual Hashing.} At this point, it is best to re-emphasize the distinction between perceptual hashing and a PPH family. Perceptual hashing seeks to determine whether two images are perceptually similar given their hash digests. On the other hand, property-preserving hashing (PPH)  preserves some property of its input in the hash domain. This can be any property of the input as long as it can be specified by a predicate. Thus PPH is a more general notion. As an example, one can use a PPH family to determine if two biometric templates are similar in terms of Hamming distance~\cite{boyle2018adversarially}. For the specific example of checking whether two images are within a distance threshold, the two hashing techniques are directly comparable. As mentioned in the introduction, a perceptual hashing scheme is defined as a hash function family which ensures that the hashes of two perceptually similar images are the same with high probability~\cite{zauner2010ph-thesis, hao2021pph-attack, holmgren2022nearly}. However, this probability is not required to be negligible. In contrast, a PPH scheme by definition requires that the hashes of two images within a certain distance from each other should \emph{evaluate} to the same predicate except with negligible probability. A key differential in PPH is the fact that instead of directly comparing the digests for similarity, we use a separate algorithm, i.e., $\mathsf{eval}_h$ to evaluate the embedded predicate. The reason why this is relevant is because, as mentioned in the introduction, due to the requirement of the perceptual hash being the same for perceptually similar images, there are strong lower bounds suggesting that this is not possible with non-negligible probability~\cite{o2014lowerbound-ph, holmgren2022nearly}. A PPH family circumvents this issue partly due to the use of a separate evaluation algorithm, but with the drawback that the hash digest size should be large enough to embed information about the (distance) predicate as we show next. 

\descr{Lower Bound on Compression.} 
Given a PPH for the Hamming distance predicate, Holmgren et al~\cite{holmgren2022nearly} derive a lower bound on $m$, i.e., the achievable compression or digest length. In our case the inputs are from $\mathbb{Z}_q^n$ instead of the generic set $\{0, 1\}^n$. We therefore, review their lower bound for inputs in $\mathbb{Z}_q^n$. Accordingly, we assume the PPH family is $\mathcal{H} = \{h: \mathbb{Z}_q^n \rightarrow \{0, 1\}^m\}$, and the Hamming distance predicate evaluates to $1$ if $\norm{\vtr{x} - \vtr{y}}_0 \leq t$ and $0$ otherwise, for $\vtr{x}, \vtr{y} \in \mathbb{Z}_q^n$.
%We first look at the bounds for the other three distance metrics. This paper is useful for the Euclidean bound~\cite{mitchell1966number}. 
The strategy used in~\cite{holmgren2022nearly} to get a bound on $m$ is as follows. Given a random $h \in \mathcal{H}$, and a random variable $X$ uniformly distributed over $\mathbb{Z}_q^n$, we first get the lower bound from Proposition~\ref{prop:dodis}:
\begin{equation}
\label{eq:min-ent-lower}
H_\infty(X \mid h, h(X)) \geq H_\infty(X \mid h) - m \geq H_\infty(X) - m = n \log_2 q - m.    
\end{equation}
Here, the third inequality follows since $X$ and $h$ are independently distributed, and the last equality is true because min-entropy is maximum when $X$ is uniformly distributed. Next, the task is to obtain an upper bound on $H_\infty(X \mid h, h(X))$, which would then give an upper bound on $m$ after rearranging the inequalities. The strategy used in~\cite{holmgren2022nearly} to obtain the upper bound is to find a vector $\vtr{y}$ which is at a Hamming distance exactly $t$ from $\vtr{x}$, where $\vtr{x}$ is the vector hashed under the PPH, i.e., $h(\vtr{x})$. They then exactly reconstruct $\vtr{x}$ by using the $\mathsf{eval}_h$ function of the PPH as an ``oracle''. More specifically, they first guess a vector $\vtr{y}$ which is exactly at Hamming distance $t$ from $\vtr{x}$. The number of such vectors is $\binom{n}{t}$. They then flip the bits of $\vtr{y}$ one at a time, and check whether the $\mathsf{eval}_h$ function outputs 0 or 1 on $h(\vtr{x})$ and the hash of the version of $\vtr{y}$ with one bit flipped. This uses at most $n$ applications of the $\mathsf{eval}_h$ function of the PPH. As long as the $\mathsf{eval}_h$ function has error less than $1/2n$, their algorithm can reconstruct $\vtr{x}$ with probability at least $\frac{\binom{n}{t}}{2^n} \cdot \frac{1}{2}$. Note that in their case $q = 2$. It is easy to change this for a general $q$ to $\frac{\binom{n}{t}}{q^n} \cdot \frac{1}{2}$ by assuming that the $\mathsf{eval}_h$ function has error at most $1/2qn$. Now, let $\mathcal{R}$ be their algorithm to reconstruct $\vtr{x}$, then following~\cite{holmgren2022nearly}:
\begin{align*}
    \Pr_{h, \mathbf{x}} ( \mathcal{R}(h, h(\vtr{x})) = \vtr{x}) ) &\geq \Pr_{\vtr{y}} (\norm{\vtr{x} - \vtr{y}}_0 = t)  \Pr_{\vtr{x}, \vtr{y}, h} (\mathcal{R}(h, h(\vtr{x})) = \vtr{x}) \mid \norm{\vtr{x} - \vtr{y}}_0 = t) \\
    &> \frac{\binom{n}{t}}{q^n} \cdot \frac{1}{2}
\end{align*}
Now, abusing notation by letting $\mathcal{H}$ also denote the random variable that takes on a random $h \in \mathcal{H}$, we get
\begin{align}
H_\infty(X \mid h, h(X)) 
&= -\log_2 \left(\sum_h \Pr(\mathcal{H} = h) \cdot (\max_{\vtr{x}} \Pr(X = \vtr{x} \mid h, h(\vtr{x}))\right) \nonumber\\
& \leq -\log_2 \left( \sum_h \Pr(\mathcal{H} = h) \cdot \Pr_\vtr{x} ( \mathcal{R}(h, h(\vtr{x})) = \vtr{x} \mid h, h(\vtr{x})) \right) \nonumber\\
&= -\log_2 (\Pr_{h, \vtr{x}} ( \mathcal{R}(h, h(\vtr{x})) = \vtr{x}) \label{eq:ent-upper-bound}\\
&\leq 1 + n \log_2 q - \log_2 \binom{n}{t} \nonumber
\end{align}
Combining the above with the inequality in Eq.~\eqref{eq:min-ent-lower}, and noting that $m$ is an integer, we obtain $m \geq \log_2 \binom{n}{t}$~\cite{holmgren2022nearly}. This bound is for non-robust PPH families, and hence also applies to RPPH families.

\section{\texorpdfstring{Compression Bounds for $\ell_1$-Distance PPH Family}{Compression Bounds for l1-Distance PPH Family}
}
\label{sub:l1-bound}
Assume we are given a PPH family $\mathcal{H} = \{h: \mathbb{Z}_q^n \rightarrow \{0, 1\}^m\}$ for the $\ell_1$-distance predicate of Eq.~\eqref{eq:l1-distance-pred}.
% \begin{equation}
% \label{eq:l1-distance-pred}
% P(\vtr{x}, \vtr{y}) = \begin{cases}
%     1, &\text{ if } \norm{\vtr{x} - \vtr{y}}_1 \leq t, \\
%     0, & \text{ otherwise}
% \end{cases}    
% \end{equation}
%with $\vtr{x}, \vtr{y} \in \mathbb{Z}_q^n$.
We are interested in finding a lower bound on $m$ similar to the one for the Hamming distance predicate (Section~\ref{sub:pph-ham-compress}). Unfortunately, the strategy used in~\cite{holmgren2022nearly} does not work for the $\ell_1$ distance. The main reason being that the predicate $\norm{\vtr{x} - \vtr{y}}_1 \leq t$ does not reveal enough information about $\vtr{x}$ given some vector $\vtr{y}$ which is exactly an $\ell_1$ distance of $t$ from $\vtr{x}$ to be able to recover $\vtr{x}$, apart from the fact that it can be used to sample a vector within distance $t$ from $\vtr{x}$. In the following we prove two bounds: one uses a large $t > 0.25qn$, and the other a much smaller $t \approx 0.1n$. We first prove a few results related to the $\ell_1$-norm of vectors in $\mathbb{Z}_q^n$. 

\subsection{\texorpdfstring{The $\ell_1$-Ball of Radius $t$}{The l1-Ball of Radius t}}
For any $\vtr{x} \in \mathbb{Z}_q^n$, let $B_1(\vtr{x}, q, t)$ denote the set of vectors $\vtr{y} \in \mathbb{Z}_q^n$, such that $\norm{\vtr{x} - \vtr{y}}_1 \leq t$. Since $\vtr{x}, \vtr{y} \in \mathbb{Z}_q^n$, we can assume that $t$ is an integer. Let $\vtr{0}$ denote the all 0 vector. 

% \begin{proposition}
% \label{prop:ball-contains}
% For all $\vtr{x} \in \mathbb{Z}_q^n$, we have $|B_1(\vtr{x}, q, t)| \leq  2^n|B_1(\vtr{0}, q, t)|$, where $\vtr{0}$ is the vector of all zeros.
% \end{proposition}
% \begin{proof}
% Let $\vtr{y} \in \mathbb{Z}_{q}^n$ be any vector in $B_1(\vtr{x}, q, t)$. We create a vector $\vtr{y}' \in \mathbb{Z}_q^n$ as follows. For $i \in [n]$, we set $y'_i = |x_i - y_i|$. By construction $0 \leq y_i' < q$. Furthermore, $\norm{\vtr{y}'}_1 = \norm{\vtr{x} - \vtr{y}}_1  \leq t$. Therefore, $\vtr{y}' \in B_1(\vtr{0}, q, t)$. Thus, each $\vtr{y} \in B_1(\vtr{x}, q, t)$ can be mapped to some $\vtr{y}' \in B_1(\vtr{0}, q, t)$. It is easy to see that this mapping is unique up to the sign of $x_i - y_i$ for each $i \in [n]$. Thus, this $\vtr{y}'$ is produced by at most $2^n$ vectors in $B_1(\vtr{0}, q, t)$. It follows that there can be at most $2^n |B_1(\vtr{0}, q, t) | $ vectors in $B_1(\vtr{x}, q, t)$.  
% \end{proof}
%On the other hand, we have
\begin{proposition}
\label{prop:ball-is-contained}
For all $\vtr{x} \in \mathbb{Z}_q^n$, we have $|B_1(\vtr{x}, q, t)| \geq  |B_1(\vtr{0}, q, t)|$.
\end{proposition}
\begin{proof}
Let $\vtr{y}' \in B_1(\vtr{0}, q, t)$. Construct the vector $\vtr{y}$, whose $i$th coordinate is:
\[
y_i = \begin{cases}
    x_i - y'_i, & \text{if }x_i \geq y'_i,\\
    y'_i, & \text{otherwise}
\end{cases}
\]
Then clearly $\vtr{y} \in \mathbb{Z}_q^n$. Now $\norm{\vtr{x} - \vtr{y}}_1 = \sum_{i = 1}^n |x_i - y_i|$. Consider the $i$th term in the sum. If $x_i \geq y'_i$, then
\[
|x_i - y_i| = |x_i - (x_i - y'_i)| = |y'_i|,
\]
otherwise if $x_i < y'_i$, then
\[
|x_i -y_i| = |x_i - y'_i| = y'_i - x_i \leq y'_i = |y'_i|,
\]
where the inequality follows from the fact that $x_i \geq 0$. Thus, in both cases $|x_i - y_i| \leq |y'_i|$. Therefore, $\norm{\vtr{x} - \vtr{y}}_1 \leq \norm{\vtr{y}'}_1 \leq t$. Thus, $\vtr{y} \in B_1(\vtr{x}, q, t)$. 

Next, we show that this mapping is injective. Consider two different vectors $\vtr{y}', \vtr{y}'' \in B_1(\vtr{0}, q, t)$. Assume one of the coordinates they differ in is the $i$th coordinate. Assume to the contrary that the map defined above yields a vector with the same $i$th coordinate $y_i$ for both. 
If $x_i \geq y'_i$ and $x_i \geq y''_i$, or if $x_i < y'_i$ and $x_i < y''_i$, then we get $y'_i = y''_i$ through the map above, which is a contradiction. Thus, either $x_i \geq y'_i$ and $x_i < y''_i$, or $x_i < y'_i$ and $x_i \geq y''_i$. Assume $x_i \geq y'_i$ and $x_i < y''_i$. Then, we get $y_i = x_i - y'_i = y''_i \Rightarrow x_i = y'_i + y''_i$. Since $y'_i \geq 0$, this means that $x_i \geq y''_i$. But this contradicts the fact that $x_i < y''_i$. The case when $x_i < y'_i$ and $x_i \geq y''_i$ is analogous. 
\end{proof}
%To understand why we require $q/2$ in $B_1(\vtr{0}, q/2, t)$ instead of $q$, consider $n = 1$, and assume $t = q/2$. Let $x = q/2$. Then $B_1(x, q, q/2)$ contains all $q$ ``vectors'' in $\mathbb{Z}_q$. However, $B_1(0, q, q/2)$ only contains $q/2 + 1$ vectors, i.e., those $y$ which satisfy $0 \leq y \leq q/2$.  

\begin{proposition}
\label{prop:ball-cardinality}
Let $t \geq 0$ be an integer. Then
% \[
% |B_1(\vtr{0}, q/2, t)| \geq \binom{n + t}{t} - \binom{n - 1 + t - q/2}{t - q/2}.
%  \]
 \[
\binom{n + t}{t} - \binom{n -1 + t - q}{t-q} \leq |B_1(\vtr{0}, q, t)| \leq \binom{n + t}{t}.
 \]
\end{proposition}
\begin{proof}
To simplify notation let $C(t, n)$ denote $|B_1(\vtr{0}, q, t)|$. Then,
\[
C(t, n) = \sum_{i = 0}^{q-1} C(t - i, n - 1).
\]
That is, we fix one element of the vector to $i$, and count all vectors of $(n-1)$ elements whose sum is $t-i$. The count is then complete by summing over all possible values of the fixed element in the original vector. Nore further that 
\begin{align*}
    C(t-1, n) &= \sum_{i = 0}^{q-1} C(t -1 - i, n - 1)\\
    &= C(t-1, n - 1) + C(t-2, n - 1) + \cdots  + C(t-q+1, n - 1) + C(t-q, n - 1) \\
    &= C(t, n - 1) + C(t-1, n - 1) + C(t-2, n - 1) + \cdots \\
    &+ C(t-q+1, n - 1)  + C(t-q, n - 1) - C(t, n - 1)\\
    &= C(t, n) + C(t-q, n - 1) - C(t, n - 1) \\
\Rightarrow C(t, n) &= C(t-1, n) + C(t, n-1) - C(t-q, n-1).
\end{align*}
The recurrence relation $C(t-1, n) + C(t, n-1)$ has the solution $\binom{n + t}{t}$ (see for example~\cite[\S C]{asghar2014algebraic}). Therefore, 
\[
C(t, n) = \binom{n + t}{t} - C(t-q, n-1) \leq \binom{n + t}{t},
\]
since $C(t-q, n-1) \geq 0$.
From the above, we see that
\begin{align*}
       C(t-q, n-1) &=  \binom{n -1 + t - q}{t-q} - C(t-2q, n - 2) \\
       &\leq \binom{n -1 + t - q}{t-q}.
\end{align*}
Therefore
\[
C(t, n) \geq \binom{n + t}{t} - \binom{n -1 + t - q}{t-q}
\]
\end{proof}

\begin{proposition}
\label{prop:ball-contains}
Let $t \geq 0$ be an integer. Let $q \geq 2$ be even. Let $\vtr{y}$ be the $n$-element vector $(q/2, \ldots, q/2)$. Then 
\[
|B_1(\vtr{y}, q, t)| \leq \binom{n + t + 1}{t + 1}
\]
\end{proposition}
\begin{proof}
Again, to simplify notation let $C(t, n)$ denote $|B_1(\vtr{y}, q, t)|$. Let $C'(t, n)$ denote $|B_1(\vtr{y}, q + 1, t)|$. Then it is easy to see that $C(t, n) \leq C'(t, n)$. We will work with $C'(t, n)$. First note that
\[
C'(t, n) = 1 + 2 \sum_{i = 0}^{q/2 - 1} C'(t - i, n-1).
\]
This is because fixing one element of a vector $\vtr{x} \in \mathbb{Z}_{q+1}^n$ within the $\ell_1$-ball of $\vtr{y}$ to $i$ reduces the problem to counting vectors of $n-1$ elements whose sum is $t - i$.
Without loss of generality, let us assume that this is the first element of $\vtr{x}$. Since the corresponding element in $\vtr{y}$ is fixed at $q/2$, $|i - q/2|$ ranges from $1$ to $q/2 - 1$. This counts all vectors whose first element is fixed at $i \in \{0, 1, \ldots, q/2-1\}$. The first element of $\vtr{x}$ can also take on the values $q/2 + 1$ to $q$. For each such value, there is a corresponding value between $0$ and $q/2 - 1$. Thus, we need to count such vectors twice, as the sum $C'(t - i, n-1)$ as $i$ ranges from $0$ to $q/2 - 1$. The remaining case is when $\vtr{x}$ is identical to $\vtr{y}$. which is a single vector. Now,
\begin{align*}
     C'(t-1, n)  &= 1 + 2\sum_{i = 0}^{q/2-1} C'(t -1 - i, n - 1)\\
    &= 1 + 2 \left( C'(t-1, n - 1) + C'(t-2, n - 1) + \cdots\right. \\
    &+ \left. C'(t-q/2+1, n - 1) + C'(t-q/2, n - 1) \right) \\
    &= 1 + 2 \left( C'(t, n - 1) + C'(t-1, n - 1) + C'(t-2, n - 1) + \cdots \right. \\
    &+ \left. C'(t-q/2+1, n - 1) 
    + C'(t-q/2, n - 1) - C'(t, n - 1) \right)\\
      &= \left( 1 + 2 \sum_{i = 0}^{q/2 - 1} C'(t - i, n-1) \right)  + 2C'(t-q/2, n - 1) - 2C'(t, n - 1)\\
    &= C'(t, n) + 2C'(t-q/2, n - 1) - 2C'(t, n - 1) \\
\Rightarrow C'(t, n) 
&= C'(t-1, n) + 2C'(t, n-1) - 2C'(t-q/2, n-1)\\
    &= \binom{n + t}{t} + C'(t, n-1) - 2C'(t-q/2, n-1)\\
    &\leq \binom{n + t}{t} + C'(t, n-1),
\end{align*}
where we have used the same identity for the recurrence relation as in Proposition~\ref{prop:ball-cardinality}, and the inequality follows since $C'(t-q/2, n-1) \geq 0$. The recurrence relation 
\[
C'(t, n) \leq \binom{n + t}{t} + C'(t, n-1),
\]
implies that
\[
C(t, n) \leq C'(t, n) \leq \binom{n + t}{t} + \binom{n - 1 + t}{t} + \cdots + \binom{t}{t} = \binom{n + t + 1}{t + 1},
\]
where the last equality is the so-called \emph{hockey-stick identity}. See for example~\cite[\S 5]{brualdi2012combinatorics}. 
\end{proof}

\subsection{\texorpdfstring{Bound on $m$ for Large $t$}{Bound on m for Large t}}
Our goal is to find a $\vtr{y} \in \mathbb{Z}_q^n$ such that $\norm{\vtr{x} - \vtr{y}} \leq t$, without knowing $\vtr{x}$. We assume $t = \gamma n$ for some $\gamma \ge 1$. Note that the maximum possible distance can be up to $(q-1)n \approx qn$. Thus, $\gamma$ gives the relative distance to the maximum possible distance. The strategy for finding such a $\vtr{y}$ is depicted in Figure~\ref{fig:midpoint}. Assume $n = 1$. Further assume that $q$ is even. In our experiments $q = 256$, and so this is not a limitation. Suppose $t = \gamma n = \gamma = q/2$. Then if we choose  $y = q/2$, it includes all points in $\mathbb{Z}_q$, and hence $x$ as well. Therefore, with probability 1, we find a $y$, namely $y = q/2$, which is a distance $\leq t = q/2$ from $x$. Of course, if $t$ is smaller than $q/2$, the probability decreases. We therefore find the probability over uniform random choices of $\vtr{x}$ that the vector $\vtr{y} = (q/2, \ldots, q/2)$ is within distance $t = \gamma n$ of $\vtr{x}$. 

\begin{figure}
\centering
\begin{tikzpicture}
\draw[thick,->] (0,0) -- (5.5,0);
\tkzDefPoint(0,0){origin};
\tkzLabelPoint[below=0.1cm](origin){$0$};
\node at (origin)[circle,fill,inner sep=1.5pt]{};

\tkzDefPoint(5,0){q};
\tkzLabelPoint[below=0.1cm](q){$q$};
\node at (q)[circle,draw, fill=white, inner sep=1.5pt]{};

\tkzDefPoint(2.5,0){y};
\tkzLabelPoint[below=0.1cm](y){$y=q/2$};
\node at (y)[circle, fill, inner sep=1.5pt]{};

\tkzDefPoint(1,0){t1};
\node at (t1)[]{$($};

\tkzDefPoint(4,0){t2};
\node at (t2)[]{$)$};

\tkzDefPoint(2.5,0.5){t};
\node at (t)[]{$2t$};

\draw[->] (2.8,0.5) -- (4,0.5);
\draw[->] (2.2, 0.5) -- (1,0.5);

\end{tikzpicture}
\caption{Choosing $y$ as the mid-point.}
\label{fig:midpoint}
\end{figure}

\begin{proposition}
\label{prop:exp-distance}
Let $\vtr{x}$ be a vector sampled uniformly at random from $\mathbb{Z}_q^n$. Let $\vtr{y} \in \mathbb{Z}_q^n$ be the vector each coordinate of which is $q/2$. Let $D$ denote the random variable denoting the $\ell_1$ distance between $\vtr{x}$ and $\vtr{y}$. Then $\mathbb{E}(D) = qn/4$.
%and $\mathbb{V}(D) = (q^2 + 12q + 8)n/48$. 
\end{proposition}
\begin{proof}
For $i \in [n]$, let $D_i$ be the random variable denoting the distance $|x_i - y_i|$. Then through linearity of expectation, $\mathbb{E}(D) = \sum_{i=1}^n \mathbb{E}(D_i)$. Since $0 \leq x_i < q$, we have $0 \leq D_i \leq q/2$. $D_i$ is 0 when $x_i = q/2$, and $D_i = q/2$, when $x_i = 0$. Any other value of $D_i$ has two possible choices of $x_i$. For instance $D_i = 1$, if $x_i = q/2 + 1$ or $q/2 - 1$. Therefore, since $x_i$ is uniformly distributed over $\mathbb{Z}_q$:
\begin{align*}
    \mathbb{E}(D_i) &= 0\cdot \frac{1}{q} + \frac{q}{2}\cdot \frac{1}{q} + 1 \cdot \frac{2}{q} + 2  \cdot \frac{2}{q} + \cdots +  \left( \frac{q}{2} - 1 \right)\cdot \frac{2}{q}\\
    &= \frac{1}{2} + \frac{2}{q} \left(1 + 2 + \cdots + \frac{q}{2} - 1\right)\\
    &= \frac{1}{2} + \frac{2}{q}  \frac{q}{2} \left(\frac{q}{2} - 1\right) \frac{1}{2} = \frac{q}{4}
\end{align*}
Therefore, $\mathbb{E}(D) = qn/4$. 
% Since $\vtr{x}$ is generated by generated each $x_i$ uniformly at random. We also get that $\mathbb{V}(D) = \sum_{i=1}^n \mathbb{V}(D_i)$. Therefore, again from the analysis above:
% \begin{align*}
%     \mathbb{E}(D^2_i) &= 0^2\cdot \frac{1}{q} + \left(\frac{q}{2}\right)^2\cdot \frac{1}{q} + 1^2 \cdot \frac{2}{q} + 2^2  \cdot \frac{2}{q} + \cdots +  \left( \frac{q}{2} - 1 \right)^2\cdot \frac{2}{q}\\
%     &= \left(\frac{q}{2}\right)^2\cdot \frac{1}{q} + \frac{2}{q} \left( \left( \frac{q}{2} -1 \right) \left( \frac{q}{2}\right)\left(2 \left( \frac{q}{2} -1 \right) + 1\right) \right) \frac{1}{6}\\
%     &= \frac{q}{2} + \frac{1}{6}  \frac{(q-1)(q-2)}{2}\\
%     &= \frac{(q+1)(q+2)}{12}
% \end{align*}
% Therefore,
% \begin{align*}
% \mathbb{V}(D_i) = \mathbb{E}(D^2_i) - (\mathbb{E}(D_i))^2 = \frac{(q+1)(q+2)}{12} - \frac{q^2}{16} = \frac{q^2 + 12 q + 8}{48}.
% \end{align*}
% The result follows by summing over $n$.
\end{proof}

\begin{proposition}
\label{prop:within-t-prob}
Let $D$ be the random variable as in Proposition~\ref{prop:exp-distance}. Let $t = \gamma n$ for some real $\gamma \ge 1$. Then if
\[
\gamma \geq \left(\frac{1}{4} + \frac{1}{2\sqrt{2n}}\right) q
\]
then $\Pr(D \leq t) \geq 1 - 1/e$.
\end{proposition}
\begin{proof}
\begin{align*}
    \Pr(D \leq t) &= \Pr(D - qn/4 \leq t - qn/4)\\
                &= 1 - \Pr(D- qn/4 > t - qn/4)\\
                &\geq 1 - \Pr(D -qn/4 \geq t - qn/4)\\
                &\geq 1 - \exp\left( - \frac{2(t-qn/4)^2}{\sum_{i=1}^n (q/2 - 0)^2}\right),
\end{align*}
where the last inequality follows from Hoeffding's inequality for $t > qn/4$, and the expected value and range of $D$ from Proposition~\ref{prop:exp-distance}. By plugging in $t = \gamma n > qn /4$, we see that $\gamma > q/4$. Assuming $\gamma$ to be such, we have
\begin{align*}
    \Pr(D \leq t) &\geq 1 - \exp\left( - \frac{2(\gamma n -qn/4)^2}{\sum_{i=1}^n (q/2)^2}\right)\\
    &= 1 - \exp\left( - \frac{2(4\gamma -q)^2n^2}{q^2 n}\frac{4}{16}\right)\\
    &= 1 - \exp\left( - \left( \frac{4\gamma}{q} - 1\right)^2 \frac{n}{2}\right).
\end{align*}
Now if 
\begin{equation}
\label{eq:int-alpha}
     \left( \frac{4\gamma}{q} - 1\right)^2 \geq \frac{2}{n},
\end{equation}
Then 
\begin{align*}
    \Pr(D \leq t) \geq 1 - \exp(-1) = 1 - 1/e,
\end{align*}
and we are done. The solution to Eq.~\eqref{eq:int-alpha} gives a bound on $\gamma$ as stated in the statement of the proposition.
\end{proof}

Thus, for large enough $t$, i.e., $t \approx qn/4$, the vector $\vtr{y} = (q/2, \ldots, q/2)$ is within $\ell_1$-distance $t$ from a uniformly random vector in $\mathbb{Z}_q^n$ with high probability. Assuming this to be the case, we then need to guess the vector $\vtr{x}$ whose hash digest has been provided to us. Our strategy is to simply sample a vector uniformly at random from $B_1(\vtr{y}, q, t)$. By assumption, $\vtr{x} \in B_1(\vtr{y}, q, t)$, and therefore the probability of obtaining $\vtr{x}$ will be $1/|B_1(\vtr{y}, q, t)|$. Sampling a vector uniformly at random from $B_1(\vtr{y}, q, t)$ is not straightforward. However, there are techniques to sample a vector from this set with an approximate uniform distribution. For instance, we can use the discrete hit-and-run sampler~\cite{baumert2009dhr}. This produces a distribution arbitrarily close to uniform~\cite[\S 11.2]{mitzenmacher2017probability}. This follows from the fact that $B_1(\vtr{y}, q, t)$ is a diamond centered around $\vtr{y}$. The resulting diamond can be enclosed within a cube which itself is within a cube and therefore the analysis in \cite[\S 4.2]{baumert2009dhr} means that the algorithm will produce a distribution arbitrarily close to the uniform distribution in polynomial time.  

Now consider an algorithm $\mathcal{R}$ which when given a random $h \in \mathcal{H}$ and an $h(\vtr{x})$ where $\vtr{x}$ is sampled uniformly at random from $\mathbb{Z}_q^n$, does as follows. It computes $\mathsf{eval}_h(h(\vtr{x}), h(\vtr{y}))$ with $\vtr{y} = (q/2, \ldots, q/2)$. If it outputs 1, it samples a vector uniformly at random from $B_1(\vtr{y}, q, t)$. It outputs this vector as its guess for $\vtr{x}$ and stops. Assume that $\mathsf{eval}_h$ has correctness error $\delta < \frac{1}{2}\cdot \frac{e-2}{e - 1}$. We get
\begin{align*}
    \Pr_{h, \vtr{x}} ( \mathcal{R}(h, h(\vtr{x}) = \vtr{x})) &\geq \Pr_{h, \vtr{x}} ( \mathcal{R}(h, h(\vtr{x}) = \vtr{x}) \mid \norm{\vtr{x} - \vtr{y}}_1 \leq t )  \Pr_{\vtr{x}}(\norm{\vtr{x} - \vtr{y}}_1 \leq t)\\
    &> \left(1-  \frac{1}{2}\cdot \frac{e-2}{e - 1}\right) \left(1 - \frac{1}{e}\right)  \frac{1}{|B_1(\vtr{y}, q, t)|}\\
    &\geq \frac{2(e-1) - (e - 2)}{2(e-1)} \frac{(e-1)}{e} \frac{1}{\binom{n+t + 1}{t + 1}}\\
    &= \frac{1}{2}\frac{1}{\binom{n+t+1}{t+1}},
\end{align*}
where we have used Proposition~\ref{prop:ball-contains}. Now using Eq~\eqref{eq:ent-upper-bound}, we get
\begin{align*}
 H_\infty(X \mid h, h(X)) &\leq -\log_2 \left(\Pr_{h, \vtr{x} } ( \mathcal{R}(h, h(\vtr{x}) = \vtr{x}) \right) \\
 &< 1 + \log_2 \binom{n + t + 1}{t + 1}. 
\end{align*}
Combining the above with Eq.~\eqref{eq:min-ent-lower}, we get
\begin{align}
    n\log_2 q - \log_2 \binom{n + t + 1}{t + 1} - 1 < m \nonumber\\
\Rightarrow m \geq n\log_2 q - \log_2 \binom{n + t + 1}{t + 1} \label{eq:m-l1},
\end{align}
where the last inequality follows from the fact that $m$ is an integer. This implies that substantial compression is possible if $t$ is large, i.e., $t > 0.25qn$. The compression rates for various image sizes $n$ and $t$ as computed through Proposition~\ref{prop:within-t-prob} are shown in Table~\ref{tab:compression-rates}. Even though high compression is possible, this value of $t$ is too large for our application where it may produce a high false positive rate.

\begin{table}[]
    \centering
\ifisfullpaper
    \resizebox{\columnwidth}{!}{
\fi
    \begin{tabular}{c|c|c|c|c|c|c|c}
         Nature & Regime & Color & $t$ & $n$ & Baseline & $m$ & Compression\\
         \hline\hline
         Bound & Large $t$ & RGB & 154918 & $28 \times 28 \times 3$ & 18816 & 1194 & 6.346 \\
         &  & & 796466 & $64 \times 64 \times 3$ & 98304 & 6495 & 6.607 \\ 
         &  & & 3165795 & $128 \times 128 \times 3$ & 393216 & 26403 & 6.715\\
         &  & & 9668908 & $224 \times 224 \times 3$ & 1204224 & 81428 & 6.762\\
         \hline
         Bound & Large $t$ & Gray- & 52711 & $28 \times 28$ & 6272 & 378 & 6.027 \\
         &  & scale & 267937 & $64 \times 64$ & 32768 & 2115 & 6.454 \\
         &  & & 1060162 & $128 \times 128$ & 131072 & 8697 & 6.635 \\
         &  & & 3231539 & $224 \times 224$ & 401408 & 26957 & 6.716\\
         \hline\hline
        Bound & Small $t$ & RGB & 235 & $28 \times 28 \times 3$ & 18816 & 882 & 4.688 \\
        &  & & 1228 & $64 \times 64 \times 3$ & 98304 & 5890 & 5.992 \\
        &  & & 4915 & $128 \times 128 \times 3$ & 393216 & 23741 & 6.038 \\
        &  & & 15052 & $224 \times 224 \times 3$ & 1204224 & 72754 & 6.042\\
        \hline
         Bound & Small $t$ & Gray- & 78 & $28 \times 28$ & 6272 & 883 & 14.078 \\
         &  & scale & 409 & $64 \times 64$ & 32768 & 1828 & 5.579 \\
         &  & & 1638  & $128 \times 128$ & 131072 & 7881 & 6.013 \\
         &  & & 5017 & $224 \times 224$ & 401408 & 24235 & 6.037 \\
         \hline\hline
         Our & Small $t$ & RGB & 235 & $28 \times 28 \times 3$ & 18816 & 2631 & 13.983 \\
         Scheme &  & & 1228 & $64 \times 64 \times 3$  & 98304 & 16682 & 16.970\\
         &  & & 4915 & $128 \times 128 \times 3$ & 393216 & 76600 & 19.480\\
         &  & & 15052 & $224 \times 224 \times 3$ & 1204224 & 258889 & 21.498\\
         \hline 
         Our & Small $t$ & Gray- & 78 & $28 \times 28$ & 6272 & 749 & 11.942 \\
         Scheme & & scale & 409 & $64 \times 64$ & 32768 & 4908 & 14.978 \\
         & & & 1638 & $128 \times 128$ & 131072 & 22932 & 17.496 \\
         & & & 5017 & $224 \times 224$ & 401408 & 78338 & 19.516\\
         \hline
% RGB: False
% scheme: small
% ['t', 'n', 'm_b', 'm', 'comp']
% [78, 784, 6272, 883, 14.07844387755102]
% [409, 4096, 32768, 1828, 5.57861328125]
% [1638, 16384, 131072, 7881, 6.012725830078125]
% [5017, 50176, 401408, 24235, 6.037498007015306]
% RGB: False
% scheme: ours
% ['t', 'n', 'm_b', 'm', 'comp']
% [78, 784, 6272, 749, 11.941964285714286]
% [409, 4096, 32768, 4908, 14.97802734375]
% [1638, 16384, 131072, 22932, 17.4957275390625]
% [5017, 50176, 401408, 78338, 19.51580436862245]
    \end{tabular}
\ifisfullpaper
    }
\fi
        \caption{Lower bounds on compression achievable through a PPH with small $t \approx 0.1 n$ and large $t \approx 0.25qn$, and the actual compression through our scheme. The column labeled ``Compression'' is the percentage compression with respect to the baseline $n \log_2 q$. Here $q = 256$.}
    \label{tab:compression-rates}
\end{table}

%For our application of preventing adversarial attacks, a value of $\alpha$ slightly over $q/4$ suffices. This means that the adversary would be forced to change over 25\% of pixel values.

\subsection{\texorpdfstring{Bound on $m$ for Small $t$}{Bound on m for Small t}}
When $t$ is much smaller, say $t \approx 0.1 n$, there does not appear to be a better algorithm than random guess to find $\vtr{x}$. Namely let $\mathcal{R}$ be an algorithm which when given a random $h \in \mathcal{H}$ and $h(\vtr{x})$ for some unknown $\vtr{x}$, samples a $\vtr{y}$ uniformly at random from $\mathbb{Z}_q^n$ as its guess for $\vtr{x}$. Assume that $\mathsf{eval}_h$ has correctness error $\delta < \frac{1}{2}$. Then,
\[
\Pr_{h, \vtr{x}} ( \mathcal{R}(h, h(\vtr{x})) > \frac{1}{2}\frac{|B_1(\vtr{x}, q, t)|}{q^n} 
\]
Now again using Eq~\eqref{eq:ent-upper-bound}, we get
\begin{align}
 H_\infty(X \mid h, h(X)) &\leq -\log_2 \left(\Pr_{h, \vtr{x} } ( \mathcal{R}(h, h(\vtr{x}) = \vtr{x}) \right) \nonumber \\
 &< 1 - \log_2 |B_1(\vtr{x}, q, t)| + n \log_2 q. \nonumber 
\end{align}
Combining the above with Eq.~\eqref{eq:min-ent-lower}, we get
\begin{align}
    m & >n\log_2 q - n \log_2 q + \log_2 |B_1(\vtr{x}, q, t)| - 1  \nonumber\\
\Rightarrow m &\geq \log_2 |B_1(\vtr{x}, q, t)| \nonumber. 
\end{align}
Now from Propositions~\ref{prop:ball-cardinality} and ~\ref{prop:ball-is-contained}, we have
\[
    |B_1(\vtr{x}, q, t)| \geq \binom{n + t}{t} - \binom{n -1 + t - q}{t-q} 
\]
We next use the following result. 
\begin{proposition}
\label{prop:binom-subtract-bound}
If $q \geq 4$ and $q-1< t < 2.5n - 2.5$, then  
\[
\binom{n + t}{t} - \binom{n -1 + t - q}{t-q} \geq \left( \frac{n-1 + t}{t} \right)^{q-1}  \binom{n -1 + t - q}{t-q}
\]
\end{proposition}
\begin{proof}
First note that 
\[
\binom{n + t}{t} = \frac{(n+t)!}{t!n!} = \frac{n+t}{n}\frac{(n-1+t)!}{t!(n-1)!} \geq \frac{(n-1+t)!}{t!(n-1)!} = \binom{n-1+t}{t}
\]
Thus,
\begin{align}
    |B_1(\vtr{x}, q, t)| &\geq \binom{n-1+t}{t} - \binom{n -1 + t - q}{t-q}\nonumber\\
    &= \left( \frac{n-1+t}{t}\right) \left( \frac{n-1+t - 1}{t-1}\right) \cdots \nonumber \\
    &\left( \frac{n-1+t - q+1}{t-q + 1}\right) \binom{n -1 + t - q}{t-q} - \binom{n -1 + t - q}{t-q}\nonumber\\
    &\geq \left( \frac{n-1+t}{t}\right)^q \binom{n -1 + t - q}{t-q} - \binom{n -1 + t - q}{t-q}\label{eq:binom-inequality}\\
    &= \left(\left( \frac{n-1+t}{t}\right)^q - 1 \right) \binom{n -1 + t - q}{t-q} \nonumber
\end{align}
where we have used the fact that for $t > q-1$, and all $0 \leq i \leq q-1$ we have
\[
\left( \frac{n-1+t - i}{t-i}\right) \geq \left( \frac{n-1+t}{t}\right). 
\]
Note that $t > q-1$ is true in our case. 
% Now if $n \geq 10$, we see that 
% \[
% \left( \frac{n-1 + t}{t}\right)^q > 2,
% \]
% as shown in the analysis in the next section. 
Now, assume $q \geq 4$ and that $(n-1 + t)/t > 1.4$ which implies that $t < 2.5(n-1)$. Then we can apply  Proposition~\ref{prop:aq-result} from Appendix~\ref{app:useful} in the previous result, and get
\begin{align*}
    |B_1(\vtr{x}, q, t)| \geq \left( \frac{n-1 + t}{t} \right)^{q-1}  \binom{n -1 + t - q}{t-q}
\end{align*}
\end{proof}

Thus
\begin{align}
    m &\geq \log_2 |B_1(\vtr{x}, q, t)| \nonumber \\
    & \geq \log_2 \left( \left( \frac{n-1 + t}{t} \right)^{q-1}  \binom{n -1 + t - q}{t-q} \right) \nonumber \\
    &= (q - 1) \log_2 \left(1 + \frac{n-1}{t} \right) + \log_2 \binom{n -1 + t - q}{t-q}  \label{eq:m-l1-constant-t}
\end{align}
Table~\ref{tab:compression-rates} shows the lower bound of $m$ through Eq.~\ref{eq:m-l1-constant-t}. Even with a smaller $t$ considerable compression is possible in theory. 

\subsection{Feasibility of List Decoding}
\label{subsec:feasibility}
Some RPPH schemes for Hamming distance are based on error-correcting codes; in particular, syndrome decoding~\cite{holmgren2022nearly}. We are interested in knowing whether syndrome decoding is feasible for $\ell_1$-distance predicates. Consider the original image $\vtr{x}$, and a candidate image $\vtr{y}$ (assume binary images for now). The construction from~\cite{holmgren2022nearly} takes the syndromes of $\vtr{x}$ and $\vtr{y}$, and finds a list of errors of Hamming weight at most $t$. Due to the linearity of syndromes, if there is an error vector in the list which matches the difference of the syndromes, then $\vtr{y}$ is within Hamming distance $t$ of $\vtr{x}$. For syndrome list decoding to be efficient, the size of the list of errors should be polynomial in $n$, which itself is a polynomial in the security parameter $\lambda$. From Fact 2.9 in~\cite{holmgren2022nearly} syndrome list decoding is efficient if and only if list decoding is efficient. 

Moving to our case, treating the target image $\vtr{x}$ as a codeword, in light of the above, for a similar procedure to be efficient we need to determine the number of vectors within the $\ell_1$-ball of $\vtr{x}$. If this size is big, then list decoding is not an efficient solution. We now estimate this size.
\begin{proposition}
\label{prop:l1-ball-too-big}
Let $B_1(\vtr{x}, q, t)$ be the $\ell_1$-ball around $\vtr{x} \in \mathbb{Z}_q^n$, with $t > q - 1$. Then 
\begin{equation}
\label{eq:l1-ball-size}
    |B_1(\vtr{x}, q, t)| > \left(1 + \frac{n-1}{t-q}\right)^{t-q}
\end{equation}
\end{proposition}
\begin{proof}
%Our first result is the other side of the coin of Proposition~\ref{prop:ball-contains}. 
From Eq.~\ref{eq:binom-inequality} in the proof of Proposition~\ref{prop:binom-subtract-bound} we have
%from Proposition~\ref{prop:ball-cardinality}, we see that
% \[
%    |B_1(\vtr{x}, q, t)| \geq \binom{n + t}{t} - \binom{n -1 + t - q}{t-q}
% \]
%Now 
% \[
% \binom{n + t}{t} = \frac{(n+t)!}{t!n!} = \frac{n+t}{n}\frac{(n-1+t)!}{t!(n-1)!} \geq \frac{(n-1+t)!}{t!(n-1)!} = \binom{n-1+t}{t}
% \]
\[
    |B_1(\vtr{x}, q, t)| \geq \left( \frac{n-1+t}{t}\right)^q \binom{n -1 + t - q}{t-q} - \binom{n -1 + t - q}{t-q}
\]
which holds for $t > q-1$, which is true by assumption. %Recall that $\alpha > q/4$, and so $t = \alpha n > qn/4$, and for any $n \ge 4$, this inequality holds. 
Now,
\begin{align*}
    \left( \frac{n-1+t}{t}\right)^q &= \left( 1 + \frac{n-1}{t}\right)^q \\
    &> \left( 1 + \frac{n-1}{qn}\right)^q,\\
    &\geq \left( 1 + \frac{0.9}{q}\right)^q,
\end{align*}
where we have used the fact that $t \leq (q-1)n < qn$ and $n \geq 10$, which is most likely to be the case. From Proposition~\ref{prop:increasing-function-q} in Appendix~\ref{app:useful}, $(1 + 0.9/q)^q$ is an increasing function of $q$, with $q \geq 2$. Therefore, since $q \geq 2$, we have for all $q \geq 2$:
\[
\left( 1 + \frac{0.9}{q}\right)^q \geq \left( 1 + \frac{0.9}{2}\right)^2  = 2.1025 > 2.
\]
Thus,
\[
    |B_1(\vtr{x}, q, t)| > 2 \binom{n -1 + t - q}{t-q} - \binom{n -1 + t - q}{t-q} = \binom{n -1 + t - q}{t-q}
\]
Now using the fact that $\binom{a}{b} \geq \left(\frac{a}{b}\right)^b$, we get:
\begin{align*}
        |B_1(\vtr{x}, q, t)| &> \binom{n -1 + t - q}{t-q} \geq \left( \frac{n -1 + t - q}{t-q} \right)^{t-q} \\
        &= \left(1 + \frac{n-1}{t-q}\right)^{t-q}
\end{align*}
\end{proof}

The lower bound from Eq.~\eqref{eq:l1-ball-size} is plotted in Figure~\ref{fig:ball-size-t} for grayscale images of various sizes with $t$ ranging from $257$ (i.e., from $q+1$) to 266. Even for these extremely small values of $t$, we see that $|B_1(\vtr{x}, q, t)|$ contains a large number of vectors. E.g., a grayscale image of size $28 \times 28$ contains at least $2^{60}$ possible vectors within its $\ell_1$-ball of radius $t = 266$. Thus, syndrome list decoding is not feasible in our case. We instead need an approach which  directly checks if a given codeword is within distance $t$ of the target codeword. 
\begin{figure}[!ht]
    \centering
    \includegraphics[scale=0.4]{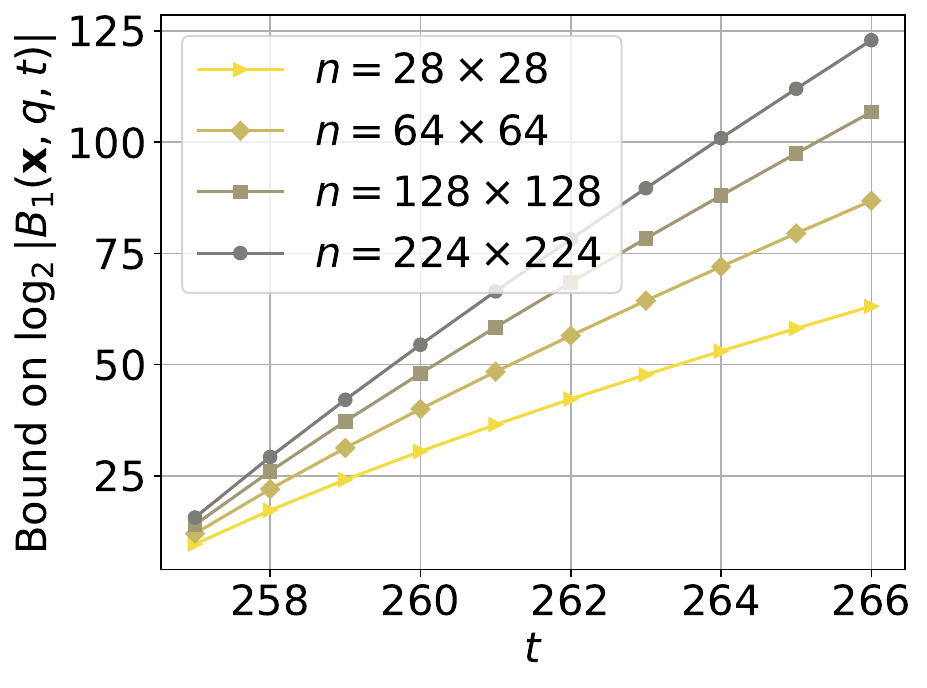}
    \caption{The lower bound in logarithmic scale of the number of images that lie within $\ell_1$-distance $t$ of a given image. The list quickly becomes huge even for such small values of $t$.}
    \label{fig:ball-size-t}
\end{figure}

% Now putting $t = \alpha n$, we get
% \[
% |B_1(\vtr{x}, q, t)| > \left(1 + \frac{n-1}{\alpha n-q}\right)^{\alpha n-q}. 
% \]
% The limit of the above function as $n \rightarrow \infty$ is $e^{n-1}$. Although, this is an asymptotic result, it turns out that the lower limit on the size is very huge for examples that we expect in our scenario. Recall that we need $\alpha$ at least $q/4$, thus we may plug in $\alpha = 0.25q$ in the above to get estimates. For the binary $28 \times 28$ images ($n = 784$) from Example~\ref{ex:images}, we get the above estimate as $\approx e^{429.5}$, with the grayscale images in the same example, we get $\approx e^{776.9}$ and with the RGB images we get $\approx e^{2332.79}$. The deviation from the limit is most probably due to the limitation of Python in dealing with large numbers. In any case, this analysis shows that list decoding is not going to be efficient for our scenario.

% \has{What happens if $t$ is much smaller? Say a constant? I believe ths size would still be too large. Best to plot this empirically}

%Thus a more direct approach whereby we can check if a given codeword is within distance $t$ of the target codeword would be more feasible. The techniques shown in the paper~\cite{tallinil1codes} look promising, and we discuss them next.

\section{\texorpdfstring{Connection to $\ell_1$-Distance Error Correcting Codes}{
Connection to l1-Distance Error Correcting Codes}}
Tallini and Rose~\cite{tallinil1codes} show a generic error correcting code for the $\ell_1$ distance, which we modify to use as a property-preserving hashing (PPH) family. To be precise, their scheme is based on the asymmetric $\ell_1$-distance. To understand the scheme, we introduce some notations. Let $\vtr{x}, \vtr{y} \in \mathbb{Z}_q^n$. For $x, y \in \mathbb{Z}_q$, define:
\[
x \dotdiv y = \max\{0, x - y\}
\]
This definition is extended element-wise to $\vtr{x} \dotdiv \vtr{y}$. Now note that 
\begin{proposition}[\cite{tallinil1codes}] For any $\vtr{x}, \vtr{y} \in \mathbb{Z}_q^n$
\label{prop:dotdiv-identity}
\[
\vtr{x} + (\vtr{y} \dotdiv \vtr{x}) = \vtr{y} + (\vtr{x} \dotdiv \vtr{y}).
\]
Furthermore,
\[
\norm{\vtr{x} - \vtr{y}}_1 = \norm{\vtr{y} \dotdiv \vtr{x}}_1 + \norm{\vtr{x} \dotdiv \vtr{y}}_1
\]
\end{proposition}
\begin{proof}
Take the $i$th element. If $x_i > y_i$, then the left hand side is $x_i + 0 = x_i$. And the right hand side is $y_i + x_i - y_i = x_i$. On the other hand if $x_i \leq y_i$, then the left hand side is $x_i + y_i - x_i = y_i$. And the right hand side is $y_i + 0 = y_i$.

For the second part, consider the $i$th summand in computing the $\ell_1$ distance. If $x_i > y_i$ then the $i$th summand of $\norm{\vtr{x} - \vtr{y}}_1$ is $|x_i - y_i| = x_i - y_i$. The $i$th summands in $\norm{\vtr{y} \dotdiv \vtr{x}}_1$ and $\norm{\vtr{x} \dotdiv \vtr{y}}_1$ are $0$ and $x_i - y_i$, respectively. Next assume $x_i \leq y_i$. Then the $i$th summand of $\norm{\vtr{x} - \vtr{y}}_1$ is $|x_i - y_i| = y_i - x_i$. The $i$th summands in $\norm{\vtr{y} \dotdiv \vtr{x}}_1$ and $\norm{\vtr{x} \dotdiv \vtr{y}}_1$ in this case are $y_i - x_i$ and $0$, respectively.
\end{proof}

\begin{example}
\label{ex:dotdiv}
Let $\vtr{x} = (2, 1, 0, 4)$ and $\vtr{y} = (3, 0, 1, 4)$. Then $\vtr{y} \dotdiv \vtr{x} = (1, 0, 1, 0)$ and $\vtr{x} \dotdiv \vtr{y} = (0, 1, 0, 0)$. Thus,
\begin{align*}
    \vtr{x} + (\vtr{y} \dotdiv \vtr{x}) &= (2, 1, 0, 4) + (1, 0, 1, 0) \\
    &= (3, 1, 1, 4)\\
    &= (3, 0, 1, 4) + (0, 1, 0, 0) = \vtr{y} + (\vtr{x} \dotdiv \vtr{y})
\end{align*}
Moreover, $\norm{\vtr{x} - \vtr{y}}_1 = 3$ and 
$\norm{\vtr{y} \dotdiv \vtr{x}}_1 + \norm{\vtr{x} \dotdiv \vtr{y}}_1 = 2 + 1 = 3$.
\qed
\end{example}
This operation is nothing but set difference if we consider $\vtr{x}$ and $\vtr{y}$ as multisets. Next we define polynomials associated with a vector in $\mathbb{Z}_q^n$ in a manner slightly different from~\cite{tallinil1codes}. Note that the goal in~\cite{tallinil1codes} is to construct error-correcting codes through which we could recover the original codeword $\vtr{x}$ given a received codeword $\vtr{y}$. In our case, we only need to find whether the received codeword $\vtr{y}$, i.e., the image, is within a certain $\ell_1$-distance away from $\vtr{x}$, a database image. Thus, we are not interested in recovering the original image $\vtr{x}$, from which $\vtr{y}$ may have been adversarially created. Let $\mathbb{F}$ be a finite field with $|\mathbb{F}| > n$. We shall assume $\mathbb{F} = \mathbb{Z}_p$, where $p > n$ is a prime. Thus $|\mathbb{Z}_p| \geq n + 1$. Let $\vtr{a} = (a_1, \ldots, a_n)$ be a vector composed of $n$ distinct elements from $\mathbb{Z}_p - \{0\}$. The polynomial associated with $\vtr{x} \in \mathbb{Z}_q^n$ is defined as
\begin{equation}
\label{eq:sigma-x}
\sigma_\vtr{x}(z) = \prod_{i=1}^n (1 - a_iz)^{x_i}
\end{equation}
Note that this is a polynomial in $\mathbb{F}[Z]$, and the coefficient operations are in the field $\mathbb{F} = \mathbb{Z}_p$. We shall informally refer to this polynomial as the \emph{$\sigma$-polynomial}. The following is stated without proof in~\cite{tallinil1codes}.

\begin{proposition}
\label{prop:sigma-x-facts}
Let $\sigma_\vtr{x} \in \mathbb{F}[Z]$ be as defined in Eq.~\eqref{eq:sigma-x}. Then $\text{deg}(\sigma_\vtr{x}) = \norm{\vtr{x}}_1$. 
\end{proposition}
\begin{proof}
First note that $\text{deg}(\sigma_\vtr{x}) \leq \sum_{i=1}^n x_i = \norm{\vtr{x}}_1$. Furthermore, the coefficient of $z^{\norm{\vtr{x}}_1}$ is given by~\cite[\S 8.6, p. 244]{toecc-book}:
\[
(-1)^{\norm{\vtr{x}}_1} \prod_{i=1}^n a_i^{x_i}.
\]
Since $a_i$ are non-zero elements of $\mathbb{F}$, it follows that the coefficient of $z^{\norm{\vtr{x}}_1}$ is non-zero. 
\end{proof}
%As a matter of fact, the coefficients of all powers of $z$ in Eq.~\eqref{eq:sigma-x} can be obtained systematically by ordering all the $a_i$'s with multiplicity $x_i$'s using the expressions given in~\cite[\S 8.6, p. 244]{toecc-book} as we discuss later in Section~\ref{subsec:sigma-polynomials-reveal}. For the current discourse, this is not required. 
The roots of $\sigma_\vtr{x}$ are $a_i^{-1} \in \mathbb{F}$ with multiplicity $x_i$ as can easily be seen. For any two polynomials in $\mathbb{F}[Z]$ define their greatest common divisor (gcd) as the zero or monic polynomial $d \in \mathbb{F}[Z]$ which divides both, and every other common divisor of the two polynomials divides $d$~\cite[\S 16.3]{shoup-nt-book}. The following \emph{key equation} is proved in~\cite{tallinil1codes}, with the proof reproduced here for completeness.
\begin{theorem}
\label{theo:key-equation}
Let $\vtr{x}, \vtr{y} \in \mathbb{Z}_q^n$. Then,
\[
\sigma_\vtr{x} \sigma_{\vtr{y} \dotdiv \vtr{x}}= \sigma_\vtr{y} \sigma_{\vtr{x} \dotdiv \vtr{y}}
\]
Furthermore,
\[
\gcd(\sigma_{\vtr{y} \dotdiv \vtr{x}}, \sigma_{\vtr{x} \dotdiv \vtr{y}}) = 1.
\]
\end{theorem}
\begin{proof}
From the definition of the $\sigma$-polynomials in Eq.~\eqref{eq:sigma-x}, together with Proposition~\ref{prop:dotdiv-identity}, we have:
\begin{align*}
     \sigma_\vtr{x}(z) \sigma_{\vtr{y} \dotdiv \vtr{x}} (z) &= \left( \prod_{i=1}^n (1 - a_iz)^{x_i} \right) \left( \prod_{i=1}^n (1 - a_iz)^{y_i \dotdiv x_i} \right) \\
     &= \prod_{i=1}^n (1 - a_iz)^{x_i + (y_i \dotdiv x_i)} \\
     &= \prod_{i=1}^n (1 - a_iz)^{y_i + (x_i \dotdiv y_i)} \\
     &= \left( \prod_{i=1}^n (1 - a_iz)^{y_i} \right) \left( \prod_{i=1}^n (1 - a_iz)^{x_i \dotdiv y_i} \right)\\
     &= \sigma_\vtr{y}(z) \sigma_{\vtr{x} \dotdiv \vtr{y}}(z)
\end{align*}
For the second part of the theorem, first note that $1 - a_i z $ does not divide $1 - a_j z$ for $i \neq j$. To see this, divide  $1 - a_j z$ by $1 - a_i z$. We get the remainder $1 - a_i^{-1} a_j$. For this to be 0, we should have $a_i = a_j$, which is not possible as all the $a_i$'s are distinct. Let $\gamma$ be any common factor of $\sigma_{\vtr{y} \dotdiv \vtr{x}} $. Since Eq.~\eqref{eq:sigma-x} is the factorization of $\sigma_{\vtr{y} \dotdiv \vtr{x}}$ into irreducible monic polynomials, $\gamma$ must have the factor $1 - a_iz$ for some $i \in [n]$. But for this to be in $\sigma_{\vtr{y} \dotdiv \vtr{x}}$ we must have that $y_i \dotdiv x > 0$, which implies that $y_i > x_i$. Therefore, $x_i \dotdiv y_i = 0$, and hence the term $1 - a_i z$ is absent in the product form of $\sigma_{\vtr{x} \dotdiv \vtr{y}}$. Since $1 - a_i z$ does not divide any other term in 
$\sigma_{\vtr{x} \dotdiv \vtr{y}}$ as established above, we have that $\gamma$ cannot be a common factor of $\sigma_{\vtr{y} \dotdiv \vtr{x}}$ and $ \sigma_{\vtr{x} \dotdiv \vtr{y}}$. 
\end{proof}

Consider the polynomial $z^{t+1} \in \mathbb{F}[Z]$. We have:
\begin{proposition}
\label{prop:gcd}
Let $\vtr{x} \in \mathbb{Z}_q^n$. Then $\gcd (\sigma_\vtr{x}, z^{t+1}) = 1$.
\end{proposition}
\begin{proof}
First note that $\sigma_\vtr{x}(0) = 1$, i.e., the constant term of $\sigma_\vtr{x}$ is 1. Let $m = \text{deg} (\sigma_\vtr{x})$. This polynomial can be written as:
\[
 \sigma_\vtr{x}(z) = A_{m}z^m + \cdots + A_{t+1}z^{t+1} + \cdots + A_1z + 1,
\]
where $A_i \in \mathbb{F}$. Now, all the divisors of $z^{t+1}$ are $z^i$ for $0 \leq i \leq t+1$. Pick any $z^i$ with $i > 0$. Then dividing $\sigma_\vtr{x}$ by $z^i$ leaves the remainder $A_{i-1}z^{i-1} + \cdots + A_1 z +  1$, where $A_0 = 1$. This is non-zero regardless of the $A_i$'s. Therefore, the gcd is 1. 
\end{proof}

We use the following results related to the extended Euclidean algorithm (EEA) whose proofs can be found in~\cite[\S 12.8]{toecc-book}.

\begin{theorem}[\cite{toecc-book}]
\label{the:eea}
Let $r_{0}(z)$ and $r_{-1}(z)$ be polynomials with $\text{deg}(r_{0}) \leq \text{deg}(r_{-1})$ and gcd $g(z)$. Then there exist polynomials $u$ and $v$ such that
\begin{equation}
\label{eq:gcd}
u(z) r_{-1}(z) + v(z) r_{0}(z) = g(z),    
\end{equation}
with $\text{deg}(u)$ and $\text{deg}(v)$ less than $\text{deg}(r_{-1})$. Furthermore, in the $i$th round, if $r_i$ and $r_{i-1}$ are the polynomials used in the division in EEA with $i \geq 0$, then
\begin{equation}
\label{eq:eea}
    \begin{bmatrix}
    r_{i-1}(z) \\ r_{i}(z)
    \end{bmatrix}
    = (-1)^i \begin{bmatrix}
    v_{i-1}(z) & -u_{i-1}(z) \\
    - v_i (z) & u_i (z) 
     \end{bmatrix}
     \begin{bmatrix}
         r_{-1}(z) \\ r_{0}(z)
     \end{bmatrix},
\end{equation}
where
\[
\begin{matrix}
    u_{-1}(z) = 0, & u_{0}(z) = 1, \\
    v_{-1}(z) = 1,  & v_{0}(z) = 0,
\end{matrix}
\]
Moreover, we have 
\begin{itemize}
    \item $u_i$ and $v_i$ are relatively prime for all $i$,
    \item $\text{deg}(u_i) = \text{deg}(r_{-1}) - \text{deg}(r_{i-1})$,
    \item $\text{deg}(u_i) = \sum_{j = 1}^{i} \text{deg}(q_{j}) $,
    \item $\text{deg}(r_{i-1}) = \text{deg}(r_{-1}) - \sum_{j = 1}^{i} \text{deg}(q_{j})$.
\end{itemize}
\qed
\end{theorem}
We have not specified the polynomials $u_i$ and $v_i$ apart from the initial values of $i$, as their expressions are not necessary for our results. Now take the key equation in Theorem~\ref{theo:key-equation} modulo $z^{t+1}$:
\begin{align}
    \sigma_\vtr{x}(z) \sigma_{\vtr{y} \dotdiv \vtr{x}}(z) &\equiv \sigma_\vtr{y}(z) \sigma_{\vtr{x} \dotdiv \vtr{y}} (z)\pmod{z^{t+1}} \nonumber \\
    \sigma_{\vtr{y} \dotdiv \vtr{x}}(z) &\equiv 
    \sigma^{-1}_\vtr{x}(z) \sigma_\vtr{y}(z) \sigma_{\vtr{x} \dotdiv \vtr{y}} \pmod{z^{t+1}}\nonumber\\
    \sigma_{\vtr{y} \dotdiv \vtr{x}}(z) 
    &\equiv \tilde{\sigma}_{\vtr{x}, \vtr{y}}(z) \sigma_{\vtr{x} \dotdiv \vtr{y}}(z) \pmod{z^{t+1}}, \label{eq:key-equation-modulo}
\end{align}
where 
\[
\tilde{\sigma}_{\vtr{x}, \vtr{y}}(z) =  \sigma^{-1}_\vtr{x}(z) \sigma_\vtr{y}(z) \pmod{z^{t+1}}.
\]
The inverse $\sigma^{-1}_\vtr{x}(z)$ exists since $\gcd (\sigma_\vtr{x}, z^{t+1}) = 1$ from Proposition~\ref{prop:gcd}, and can be obtained via the EEA (Eq.~\eqref{eq:gcd}).

\subsection{\texorpdfstring{When the Asymmetric $\ell_1$-Distances are Less Than the Thresholds}{When the Asymmetric l1-Distance is Less Than the Thresholds}}
\label{subsec:l1-less-than-t}
We first consider the case when $\norm{\vtr{y} \dotdiv \vtr{x}}_1 \leq t_+ $ and $\norm{\vtr{x} \dotdiv \vtr{y}}_1 \leq t_-$, where $t_+$ and $t_-$ are non-negative integers satisfying $t_+ + t_- = t$. From Proposition~\ref{prop:dotdiv-identity} this means that $\norm{\vtr{y} - \vtr{x}}_1 \leq t$. Moreover, from Proposition~\ref{prop:sigma-x-facts} this implies that $\text{deg}(\sigma_{\vtr{y} \dotdiv \vtr{x}}) \leq t_+$
 and $\text{deg}(\sigma_{\vtr{x} \dotdiv \vtr{y}}) \leq t_-$. The following theorem shows that if these conditions are met then we can find the solution to Eq.~\eqref{eq:key-equation-modulo} given $\tilde{\sigma}_{\vtr{x}, \vtr{y}}$ and $z^{t+1}$, and from the solution we can exactly recover $\norm{\vtr{y} \dotdiv \vtr{x}}_1$ and $\norm{\vtr{x} \dotdiv \vtr{y}}_1$. From this, we can establish that $\norm{\vtr{y} - \vtr{x}}_1 \leq t$. This result is stated in~\cite{tallinil1codes} with the proof deferred to the full version of that paper. However, we could not find the full version of the paper. The authors did in fact state that this is based on the proof of Theorem 16 in~\cite[\S 12.8]{toecc-book}. We therefore, provide a full proof here based on this theorem.

\begin{theorem}
\label{theo:unique-sol}
Let $t_+$ and $t_-$ be nonnegative integers satisfying $t_+ + t_- = t$, for a nonnegative integer $t$. Assume $\text{deg}(\sigma_{\vtr{y} \dotdiv \vtr{x}}) \leq t_+$
 and $\text{deg}(\sigma_{\vtr{x} \dotdiv \vtr{y}}) \leq t_-$. Set $r_{-1}(z) = z^{t+1}$ and $r_0(z) = \tilde{\sigma}_{\vtr{x}, \vtr{y}}(z)$ in the EEA, and run the algorithm until reaching an $r_k(z)$ such that 
 \[
 \begin{matrix}
     \text{deg}(r_{k}) \leq t_+ & \text{and} & \text{deg}(r_{k-1}) > t_+. 
 \end{matrix}
 \]
 Set 
 \[
 \begin{matrix}
 \alpha(z) = (-1)^k r_k(z), & \qquad      \beta(z) = u_k(z) 
 \end{matrix}
 \]
 Then $\text{deg}(\alpha) = \text{deg}(\sigma_{\vtr{y} \dotdiv \vtr{x}}) $ and $\text{deg}(\beta) = \text{deg}(\sigma_{\vtr{x} \dotdiv \vtr{y}})$. 
\end{theorem}
 \begin{proof}
 Set $r_{-1}(z) = z^{t+1}$ and $r_0(z) = \tilde{\sigma}_{\vtr{x}, \vtr{y}}(z)$. Run the EEA until reaching an $r_k(z)$ such that
 \[
 \begin{matrix}
     \text{deg}(r_{k}) \leq t_+ & \text{and} & \text{deg}(r_{k-1}) > t_+. 
 \end{matrix}
 \]
 Note that this is guaranteed since we start with $\text{deg}(r_{-1}) = t+1$, and  $\text{deg}(r_{i})  < \text{deg}(r_{i-1})$  at the $i$th iteration. Furthermore, $\gcd(\tilde{\sigma}_{\vtr{x}, \vtr{y}}, z^{t+1}) = 1$. So, the degrees of $r_i$'s are decreasing and go down to 0. Similarly we start with $\text{deg}(u_{0}) = 1$, and at the $i$th iteration, we have $\text{deg}(r_{i-1}) > \text{deg}(r_{i})$, which means (from Theorem~\ref{the:eea})
 \[ 
    \text{deg}(u_{i}) = \text{deg}(r_{-1}) - \text{deg}(r_{i-1}) > \text{deg}(r_{-1}) - \text{deg}(r_{i-2}) =  \text{deg}(u_{i-1}),
\]
and hence the degrees of the $u_i$'s are increasing. From the same theorem: 
\begin{equation}
\label{eq:uk-less-than-tminus}
 \text{deg}(u_k) = \text{deg}(r_{-1}) - \text{deg}(r_{k-1}) < t + 1 - t_{+} = t_{-} + 1 \leq t_{-}   
\end{equation}
Now set 
%  $\text{deg}(r_{k-1}) \geq \text{deg}(r_{k}) + 1 = t_+ + 1$ and
% \[ 
%     \text{deg}(u_{k}) = \text{deg}(r_{-1}) - \text{deg}(r_{k-1}) \leq t + 1 - t_+ - 1 = t - t_+
% \]
 \begin{align}
     \alpha(z) &= (-1)^k r_k(z), \nonumber\\
     \beta(z) &= u_k(z) \label{eq:ab}.
 \end{align}
Thus $\text{deg}(\alpha) = \text{deg}(r_{k}) \leq t_+$, and $\text{deg}(\beta) = \text{deg}(u_{k}) \leq t_-$. From Eq.~\eqref{eq:eea}, we have
\begin{align}
    r_k(z) &= (-1)^k (-v_k(z)r_{-1}(z) + u_k(z) r_0(z)) \nonumber\\
\Rightarrow    (-1)^k r_k(z) &= -v_k(z)r_{-1}(z) + u_k(z) r_0(z) \nonumber\\
\Rightarrow   \alpha(z) &= \beta(z) \tilde{\sigma}_{\vtr{x}, \vtr{y}}(z) -v_k(z)z^{t+1}  \nonumber\\
\Rightarrow \alpha(z) &\equiv \beta(z) \tilde{\sigma}_{\vtr{x}, \vtr{y}}(z) \pmod{z^{t+1}}. \nonumber
\end{align}
Thus $(\alpha, \beta)$ is a solution of Eq.~\eqref{eq:key-equation-modulo}. We next show that if $(\alpha', \beta')$ is any other solution of Eq.~\eqref{eq:key-equation-modulo} satisfying $\text{deg}(\alpha') \leq t_+$, and $\text{deg}(\beta') \leq t_-$, and $\gcd(\alpha', \beta') = 1$ then necessarily $\text{deg}(\alpha') = \text{deg}(\alpha)$ and 
$\text{deg}(\beta') = \text{deg}(\beta)$. Since we know that $\text{deg}(\sigma_{\vtr{y} \dotdiv \vtr{x}}) \leq t_+$
 and $\text{deg}(\sigma_{\vtr{x} \dotdiv \vtr{y}}) \leq t_-$, that they satisfy Eq.~\eqref{eq:key-equation-modulo} and $\gcd(\sigma_{\vtr{y} \dotdiv \vtr{x}}, \sigma_{\vtr{x} \dotdiv \vtr{y}}) = 1$ from Theorem~\ref{theo:key-equation}, this shows that the solution through EEA, i.e., Eq.~\eqref{eq:ab}, will satisfy $\text{deg}(\alpha) = \text{deg}(\sigma_{\vtr{y} \dotdiv \vtr{x}})$ and $\text{deg}(\beta) = \text{deg}(\sigma_{\vtr{x} \dotdiv \vtr{y}})$, and we are done. So assume $(\alpha', \beta')$ is another solution. Then
%We assume that $\alpha'$ and $\beta'$ have no common factors, i.e., $\gcd(\alpha', \beta') = 1$. For if they do, then we can cancel out the common factor to obtain smaller polynomials $\alpha''$ and $\beta''$ in the above congruence with their respective degrees being less than $t_+$ and $t_-$, respectively. Now the above implies
\begin{align}
 \alpha'(z) &\equiv \beta'(z) \tilde{\sigma}_{\vtr{x}, \vtr{y}}(z) \pmod{z^{t+1}} \label{eq:ab-congruence} \\
    \alpha'(z) \beta(z) &\equiv \beta'(z)  \tilde{\sigma}_{\vtr{x}, \vtr{y}}(z) \beta(z) \pmod{z^{t+1}} \nonumber\\
    \alpha'(z) \beta(z) &\equiv \alpha(z) \beta'(z) \pmod{z^{t+1}} \nonumber
\end{align}
The degree of each side of this congruence is $\leq t_+ + t_- = t$, and hence we have 
\[
\alpha'(z) \beta(z) = \alpha(z) \beta'(z),
\]
i.e., without the modulus. Since $\alpha'$ divides both sides, we have
\[
\beta(z) = \frac{\alpha(z) \beta'(z)}{\alpha'(z)}
\]
Since $\alpha'$ and $\beta'$ are relatively prime, $\alpha'$ must divide $\alpha$.
Define:
\[
\mu(z) = \frac{\alpha(z)}{\alpha'(z)},
\]
which implies ${\beta(z)} = \mu(z) {\beta'(z)}$. From Eq.~\eqref{eq:ab} and \eqref{eq:eea} we have
\begin{align*}
    \alpha(z) &= (-1)^k r_k(z) \\
            &= -v_k(z)r_{-1}(z) + u_k(z) r_0(z) \\
            &= -v_k(z) z^{t+1} + \beta(z) \tilde{\sigma}_{\vtr{x}, \vtr{y}}(z)\\
\Rightarrow \mu(z) \alpha'(z) &= -v_k(z) z^{t+1} + \mu(z) \beta'(z) \tilde{\sigma}_{\vtr{x}, \vtr{y}}(z).
\end{align*} 
Since $(\alpha', \beta')$ is a solution we have for some polynomial $\psi$, 
\[
\alpha'(z) = \beta'(z) \tilde{\sigma}_{\vtr{x}, \vtr{y}}(z) + \psi(z) z^{t+1}
\]
Putting this in the previous equation, we get
\begin{align*}
\mu(z) \beta'(z) \tilde{\sigma}_{\vtr{x}, \vtr{y}}(z) + \mu(z) \psi(z) z^{t+1} &= -v_k(z) z^{t+1} + \mu(z) \beta'(z) \tilde{\sigma}_{\vtr{x}, \vtr{y}}(z) \\
-\mu(z) \psi(z) &= v_k(z).
\end{align*}
Thus, $\mu$ divides $v_k$. On the other hand $\mu(z) \beta'(z) = \beta(z) = u_k(z)$. Thus, $\mu$ also divides $u_k$. But $u_k$ and $v_k$ are relatively prime from Theorem~\ref{the:eea}. Thus, $\mu$ is a constant, and hence $\text{deg}(\alpha') = \text{deg}(\alpha)$ and $\text{deg}(\beta') = \text{deg}(\beta)$. 
% Now 

% \[
% \begin{matrix}
%     \text{deg}(\alpha') = \text{deg}(\gamma') + \text{deg}(\alpha'') & \text{and} & \text{deg}(\beta') = \text{deg}(\gamma') + \text{deg}(\beta'')
% \end{matrix}
% \]

% If $\gamma'$ is not a constant then 
% \[
% \begin{matrix}
%      \text{deg}(\alpha') > \text{deg}(\alpha'') = \text{deg}(\alpha) & \text{and} &  \text{deg}(\beta') > \text{deg}(\beta'') = \text{deg}(\beta) 
% \end{matrix}
% \]
% But this is not possible through EEA as $\alpha'$ should be encountered before $\alpha$ and therefore the degree of $\beta'$ should be less than $\beta$ as mentioned in the beginning of this proof. 
% %Thus, $\gamma'$ is a constant, and again we have $ \text{deg}(\alpha') = \text{deg}(\alpha)$ and $ \text{deg}(\beta') = \text{deg}(\beta)$. 
% Thus, any solution through EEA should have $\gcd(\alpha, \beta) = 1$. But that is also true of $\sigma_{\vtr{y} \dotdiv \vtr{x}}$  and $\sigma_{\vtr{x} \dotdiv \vtr{y}}$. Thus, $\alpha$ and $\beta$ are constant multiples of $\sigma_{\vtr{y} \dotdiv \vtr{x}}$  and $\sigma_{\vtr{x} \dotdiv \vtr{y}}$ respectively, and the result follows.
\end{proof}

\subsection{\texorpdfstring{When the Asymmetric $\ell_1$-Distances are More Than the Thresholds}{When the Asymmetric l1-Distance is More Than the Threshold}}
\label{subsec:l1-more-than-t}
We now consider the other case, i.e., when $\norm{\vtr{y} \dotdiv \vtr{x}}_1 > t_+ $ or $\norm{\vtr{x} \dotdiv \vtr{y}}_1 > t_-$, where $t_+$ and $t_-$ are non-negative integers satisfying $t_+ + t_- = t$. From Proposition~\ref{prop:sigma-x-facts}, equivalently, this means that 
$\text{deg}(\sigma_{\vtr{y} \dotdiv \vtr{x}}) > t_+$
 or $\text{deg}(\sigma_{\vtr{x} \dotdiv \vtr{y}}) > t_-$. The proof of Theorem~\ref{theo:unique-sol} tells us that even if the key equation (Eq.~\eqref{eq:key-equation-modulo}) is not satisfied, there will still be a solution if we run the EEA till $\text{deg}(r_k) \leq t_{+}$ and $\text{deg}(r_{k-1}) > t_{+}$, since we have $\text{deg}(u_k) \leq t_{-}$ (from Eq.~\eqref{eq:uk-less-than-tminus}) which has the same degree as $\sigma_{\vtr{x} \dotdiv \vtr{y}}$. However, if we change the condition to running the algorithm until $\text{deg}(r_k) < t_+$ and $\text{deg}(r_{k-1}) \geq t_{+}$, then we see from Theorem~\ref{the:eea} that
\[
\text{deg}(u_k) = \text{deg}(r_{-1}) - \text{deg}(r_{k-1}) \leq t + 1 - t_{+} = t_{-} + 1.
\]
Furthermore from Theorem~\ref{the:eea}, 
\begin{align*}
   \text{deg}(r_k) &= \text{deg}(r_{-1}) - \sum_{i = 1}^{k+1} \text{deg}(q_{i}) \\
    \Rightarrow \text{deg}(u_k) + \text{deg}(q_{k+1}) &= \text{deg}(r_{-1}) - \text{deg} (r_k) \\
    &> t + 1 - t_+ = t_- + 1,
\end{align*}
where we have used the fact that $\text{deg}(u_k) = \sum_{i = 1}^{k} \text{deg}(q_{i})$ from Theorem~\ref{the:eea}. Together, we get the condition:
\[
t_- + 1 \geq \text{deg}(u_k) > t_- + 1 - \text{deg}(q_{k+1})
\]
where $q_{k+1}$ is defined as
\[
r_{k-1} = q_{k+1} r_{k} + r_{k+1}.
\]
Since $\text{deg}(r_{k+1}) < \text{deg}(r_{k})$, we have
\begin{align*}
    \text{deg}(r_{k-1}) &= \text{deg}(q_{k+1}) + \text{deg}(r_{k}) \\
    \Rightarrow \text{deg}(q_{k+1}) &= \text{deg}(r_{k-1}) - \text{deg}(r_{k}) \geq 1.
\end{align*}
Thus if $ \text{deg}(q_{k+1}) = 1$ then 
\[
t_- + 1 \geq \text{deg}(u_k) > t_- \Rightarrow \text{deg}(u_k) = t_- + 1.
\]
However, $ \text{deg}(q_{k+1})$ could be greater than $1$. In general if $ \text{deg}(q_{k+1}) = \delta + 1$, where $\delta \geq 0$ is an integer, we get 
\[
t_- + 1 \geq \text{deg}(u_k) > t_- - \delta.
\]
Thus, once the EEA stops at the condition $\text{deg}(r_k) < t_+$ and $\text{deg}(r_{k-1}) \geq t_{+}$, we could flag it as a non-solution if $\text{deg}(u_k) > t_- -\delta$ for some fixed integer $\delta \geq 0$. $\delta = 0$ is guaranteed to happen. But $\delta = 1$ is less probable, $\delta = 2$ even less so and so on. But how probable is this? Let $\tau$ denote the degree of $r_{k-1}$. Then from the equation:
\[
r_{k-2} = q_{k} r_{k-1} + r_{k},
\]
we see that the polynomial on the left hand side and the one on the right hand side are of the form:
\[
A_\tau z^\tau + A_{\tau - 1} z^{\tau - 1} + \cdots + A_1 z + A_0 = B_\tau z^\tau + B_{\tau - 1} z^{\tau - 1} + \cdots + B_1 z + B_0 + r_{k}(z).
\]
Since $\text{deg}(r_{k-2}) > \text{deg}(r_{k-1})$, we have $A_\tau = B_\tau$. Now, $A_{\tau - 1} = B_{\tau - 1}$ implies that 
\[
\text{deg}(r_{k-1}) - \text{deg}(r_{k}) = \text{deg}(q_{k+1}) \geq 2,
\]
and in general $A_{\tau - 1} = B_{\tau - 1}, \ldots, A_{\tau - \delta} = B_{\tau - \delta}$ implies that 
\[
\text{deg}(r_{k-1}) - \text{deg}(r_{k}) = \text{deg}(q_{k+1}) \geq \delta + 1.
\]
The EEA starts by dividing $z^{t+1}$ by $\tilde{\sigma}_{\vtr{x}, \vtr{y}}(z)$. The coefficients of $\tilde{\sigma}_{\vtr{x}, \vtr{y}}$ are sums of products of random elements (without replacement) of $\mathbb{F}$ (due to the vector $\vtr{a}$). Thus, we can model them as random coefficients in $\mathbb{F}$. It follows that the probability can be approximated as
\begin{equation}
\label{eq:deg-q-one-over-F}
 \Pr[\text{deg}(q_{k+1}) \geq \delta + 1 ] \approx \frac{1}{|\mathbb{F}|^\delta} = \frac{1}{|\mathbb{Z}_p|^\delta} = \frac{1}{p^\delta}.   
\end{equation}
If $\mathbb{F}$ is large enough, say around 1000, then setting $\delta = 3$ suffices, as the chance of $\delta \geq 3$ is approximately one in a billion. Unfortunately, we do not have an analytical proof of this, which we leave as an open problem. However, we can show through simulations that this is a very good estimate of the probability, as shown in Figure~\ref{fig:non-sol-prob}. We have chosen very small values of $n$ and $t$, since the probabilities are already very small. The probabilities are obtained by implementing the scheme in the Python library \verb+galois+~\cite{Hostetter_Galois_2020}. For each value of $n$ we choose $p$ as a prime larger than $n$. We then sample uniformly random images $\vtr{x} \in \mathbb{Z}_q^n$, and make changes to it resulting in the vector $\vtr{y}$ such that $\norm{\vtr{y} \dotdiv \vtr{x}}_1 \geq t_+$ and $\norm{\vtr{x} \dotdiv \vtr{y}}_1 > t_-$. From the figure, we can see that across all cases the empirical probability of the left hand side of Eq.~\ref{eq:deg-q-one-over-F} is less than $p^{-\delta}$. Furthermore, the probability decreases significantly as we increase $p$, as is likely to be the case for practical values of $n$ for actual images.
 
\begin{figure*}[!ht]
    \centering
    \begin{subfigure}[t]{0.4\textwidth}
        \includegraphics[width=\linewidth]{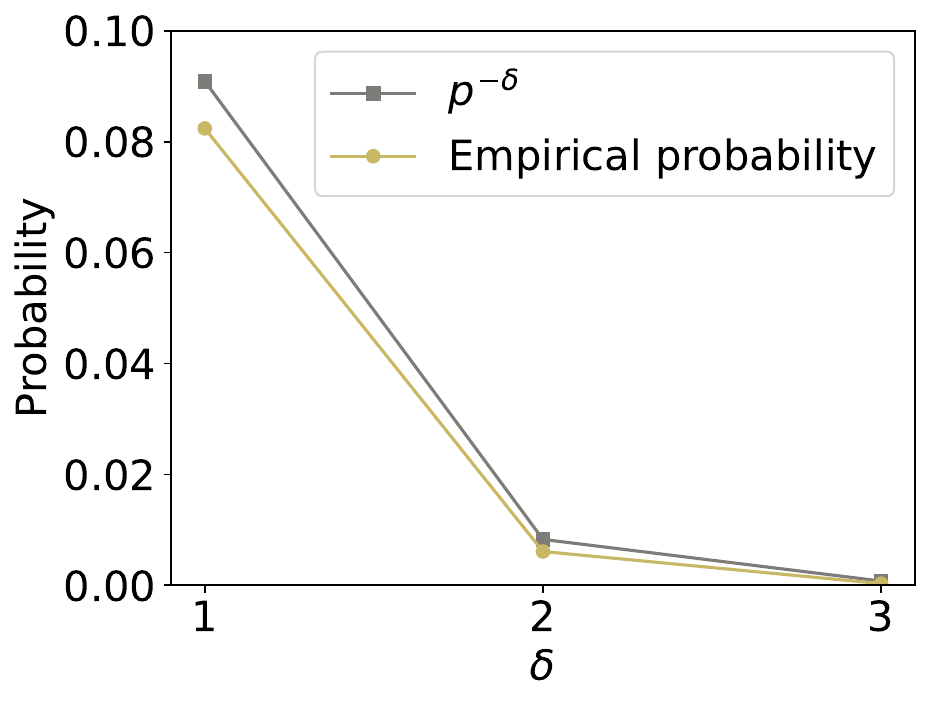}
        \caption{$ (10, 11, 10, 5, 5 - \delta$)}
        \label{subfig:non-sol-prob-p-11-n-10}
    \end{subfigure}
    \qquad
    \begin{subfigure}[t]{0.4\textwidth}
        \includegraphics[width=\linewidth]{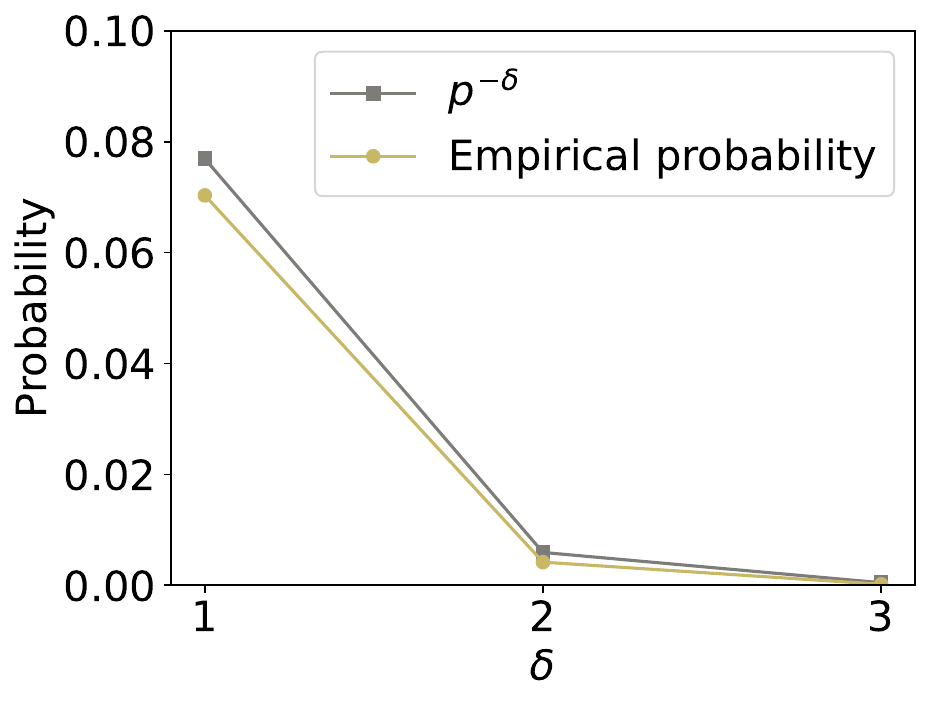}
        \caption{$(10, 13, 10, 5, 5 - \delta$)}
        \label{subfig:non-sol-prob-p-13-n-10}
    \end{subfigure}
    \qquad
    \begin{subfigure}[t]{0.4\textwidth}
        \includegraphics[width=\linewidth]{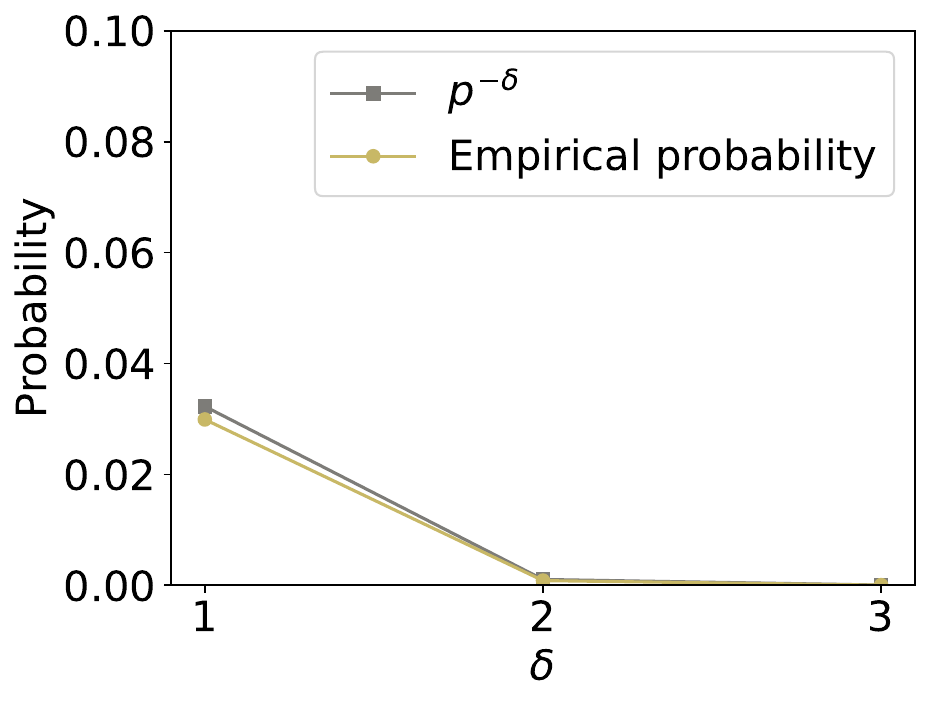}
       \caption{$(10, 31, 10, 5, 5 - \delta$)}
        \label{subfig:non-sol-prob-p-31-n-10}
    \end{subfigure}
    \qquad
    \begin{subfigure}[t]{0.4\textwidth}
        \includegraphics[width=\linewidth]{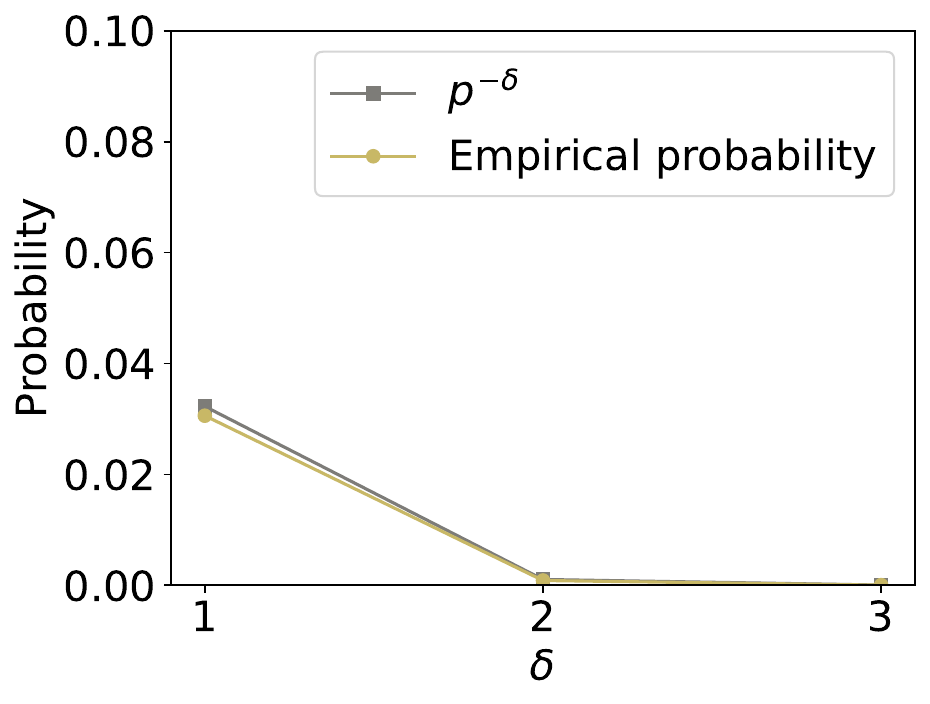}
        \caption{$(10, 31, 30, 15, 15 - \delta$)}
        \label{subfig:non-sol-prob-p-31-n-30}
    \end{subfigure}
    \caption{The probability $p^{-\delta}$ from Eq.~\eqref{eq:deg-q-one-over-F} versus the empirical probability obtained after $10^6$ runs with varying $n$ and $p > n$. We use $q = 5$ in all plots. In all cases, the empirical probability is lower than $p^{-\delta}$. Tuples are the values $(n,p,t,t_+,t_-)$.} 
    \label{fig:non-sol-prob}
\end{figure*}

\section{PPH Construction and Security Analysis}
We first precisely define the asymmetric $\ell_1$-distance predicate on images from $\mathbb{Z}_q^n$.
\begin{definition}
Let $n$ be a positive integer. Let $q \geq 2$ and $\delta \geq 0$ be integers. Let $t$ be a positive integer and let $t_+$ and $t_-$ be non-negative integers with $t = t_+ + t_-$. The two-input asymmetric $\ell_1$-distance predicate $P_{\mathsf{as}}$ is defined as
\label{def:l1-asym-pred}
\[
P_{\mathsf{as}}(\vtr{x}, \vtr{y}) =
\begin{cases}
     1, &\text{ if } \norm{\vtr{y} \dotdiv \vtr{x}}_1 < t_+ \text{ and } \norm{\vtr{x} \dotdiv \vtr{y}}_1 \leq t_- - \delta \\
    0, & \text{ otherwise}
\end{cases}
\]\qed
\end{definition}
The $(m, n)$-PPH construction for this predicate is shown in Construction~\ref{construct:pph}. We use the fact that a degree-$t$ polynomial with coefficients in $\mathbb{Z}_p$ can be represented as a vector in $\mathbb{Z}^{t+1}_p$, where the $i$th coefficient of the polynomial is the $i$th element in the vector representation. Thus, $m = (t+1)\log_2 p \approx t \log_2 n$.

\begin{algorithm}
\SetAlgorithmName{Construction}{List of constructions}
\SetAlgoLined
\DontPrintSemicolon
\SetKwInOut{Input}{Parameters}
\Input{Security parameter $\lambda$, positive integers $n = n(\lambda)$ and $q \geq 2$ for $\mathbb{Z}_q^n$,  positive integers $t = t(\lambda)$, $t_+$ and $t_-$, with $t = t_+ + t_-$, integer $\delta \geq 0$, input image $\vtr{x} \in \mathbb{Z}_q^n$.}
$\bullet \; \mathsf{samp}(1^\lambda)$:\;
\Indp 1. Set $p$ as the first prime after $n$\;
2. Generate $\vtr{a} = (a_1, \ldots, a_n)$ as a vector with $n$ distinct elements from $\mathbb{Z}_p - \{0\}$.\;
3. Output the following hash function $h$:
\[
h(\vtr{x}) = \sigma_{\vtr{x}}(z) \pmod{z^{t+1}} \in \mathbb{Z}_p^{t+1}
\]
where\;
\[
\sigma_\vtr{x}(z) = \prod_{i=1}^n (1 - a_iz)^{x_i} \; \text{ with coefficients in }\mathbb{Z}_p
\]
\Indm $\bullet \; \mathsf{eval}_h(X, Y)$: Let $X, Y \in \mathbb{Z}_p^{t+1}$.\;
\qquad \;
\Indp 1. Compute:\;
\[
\tilde{\sigma}_{\vtr{x}, \vtr{y}}(z) =  \sigma^{-1}_\vtr{x}(z) \sigma_\vtr{y}(z) \pmod{z^{t+1}}.
\]
2. Set $r_{-1}(z) = z^{t+1}$ and $r_0(z) = \tilde{\sigma}_{\vtr{x}, \vtr{y}}(z)$ in the EEA, and run the algorithm until reaching an $r_k(z)$ such that\;
 \[
 \begin{matrix}
     \text{deg}(r_{k}) < t_+ & \text{and} & \text{deg}(r_{k-1}) \geq t_+. 
 \end{matrix}
 \]
3. Output 1 if $\text{deg}(u_k) \leq t_- -\delta$, else output 0. 
\caption{An $(m, n)$-PPH for the Asymmetric $\ell_1$-Distance Predicate}
\label{construct:pph}
\end{algorithm}

\subsection{Correctness and Efficiency}
\label{subsec:correctness}
\begin{proposition}
\label{prop:scheme-is-pph}
Construction~\ref{construct:pph} is a $(1, m, n)$-PPH for the asymmetric $\ell_1$-distance predicate $P_{\mathsf{as}}$, with $m = (t + 1) \log_2 p$ and $p > n$ being a prime. 
\end{proposition}
\begin{proof}
The $1$-correctness of the PPH follows almost directly from Theorem~\ref{theo:unique-sol}. More elaborately, if $\vtr{x}$ and $\vtr{y}$ are such that $\norm{\vtr{y} \dotdiv \vtr{x}}_1 < t_+$ and $\norm{\vtr{x} \dotdiv \vtr{y}}_1 \leq t_- - \delta$, then $P_{\mathsf{as}}(\vtr{x}, \vtr{y}) = 1$. Let $\sigma_\vtr{x}$ and $\sigma_\vtr{y}$ be the corresponding $\sigma$-polynomials. From Proposition~\ref{prop:gcd}, we have $\gcd (\sigma_\vtr{x}, z^{t+1}) = 1$, and therefore the inverse $\sigma^{-1}_\vtr{x}(z)$ exists. Thus, we can obtain $\tilde{\sigma}_{\vtr{x}, \vtr{y}}(z) =  \sigma^{-1}_\vtr{x}(z) \sigma_\vtr{y}(z) \pmod{z^{t+1}}$. We can then run the EEA with inputs $\tilde{\sigma}_{\vtr{x}, \vtr{y}}$ and $z^{t+1}$. Changing the stopping condition to $\text{deg}(r_{k}) < t_+ $ and $\text{deg}(r_{k-1}) \geq t_+$ does not change the result of Theorem~\ref{theo:unique-sol} as can be easily verified. Therefore, since
%Let $t^*_+ = t_+ - 1$, $t^*_- = t_- - \delta$ and $t^* = t^*_+ + t^*_-$. Since 
$\text{deg}(\sigma_{\vtr{y} \dotdiv \vtr{x}}) < t_+$ and 
$\text{deg}(\sigma_{\vtr{x} \dotdiv \vtr{y}}) \leq t_- - \delta \leq t_-$, the theorem guarantees that the degrees of the polynomials $\alpha$ and $\beta$, obtained through the EEA, are equal to the degrees of $\text{deg}(\sigma_{\vtr{y} \dotdiv \vtr{x}})$ and $\text{deg}(\sigma_{\vtr{x} \dotdiv \vtr{y}})$, respectively. Thus $\mathsf{eval}_h$ will output 1 in this case. Hence,
\[
\Pr_{h \leftarrow \mathsf{samp}(1^\lambda)} \left[ P(x_1, x_2) \neq \mathsf{eval}_h(h(x_1), h(x_2)) \mid P(x_1, x_2) = 1 \right] = 0
\]
\end{proof}
For $0$-correctness we only have the following conjecture due to Eq.~\eqref{eq:deg-q-one-over-F}:
\begin{equation*}
%\label{eq:pph-0-correct}
\Pr_{h \leftarrow \mathsf{samp}(1^\lambda)} \left[ P(x_1, x_2) \neq \mathsf{eval}_h(h(x_1), h(x_2)) \mid P(x_1, x_2) = 0 \right] \approx {p^{-\delta}}
\end{equation*}
Asymptotically, $p^{-\delta}$ is not a negligible function of $\lambda$, as $p$ is the next prime to $n$, which itself is polynomial in $\lambda$. However, in practice, this can be made extremely small. For instance, for $n = 224 \times 224 \times 3$ and $\delta = 3$, we have $p^{-\delta} < 2^{-51}$. Since the above theorem holds for all vectors $\vtr{a}$ sampled through $\mathsf{samp}()$, we also have:

\begin{proposition}
\label{prop:scheme-is-rpph}
Construction~\ref{construct:pph} is a robust $(1, m, n)$-PPH for the asymmetric $\ell_1$-distance predicate $P_{\mathsf{as}}$, with $m = (t + 1) \log_2 p$ and $p > n$ being a prime. \qed
\end{proposition}
We do not have an equivalent conjecture to claim that our scheme is also $(0, m, n)$-RPPH, as it may be possible to find a pair of images $\vtr{x}$ and $\vtr{y}$, given an instant of the hash function $h$ from the family, i.e., the prime $p$ and vector $\vtr{a}$. Given the $\sigma$-polynomial $\sigma_{\vtr{x}}$ of an image $\vtr{x}$, even though we can trivially find a collision in the polynomial space, e.g., $z^{t+1} + \sigma_\vtr{x}$, the resulting polynomial needs to be a valid $\sigma$-polynomial of some image $\vtr{y}$ for it to be a collision in the image space. We discuss reverse-engineering the original image in Section~\ref{subsec:sigma-polynomials-reveal}, which also discusses finding collisions in the image space.  

From the theorem, the compression achieved by our scheme is $\approx t \log_2 n$. For small $t$ the compression rate is close to the lower bound as shown in Table~\ref{tab:compression-rates}. Unfortunately, for large $t$, not much compression is possible. It is unclear whether further compression is possible for robust PPH families, since the compression bounds from Section~\ref{sub:l1-bound} are for non-robust PPH families. For efficiency of the construction, we have the following theorem.
\begin{theorem}
\label{theo:pph-is-efficient}
Let $h$ be a hash function from Construction~\ref{construct:pph}. Then $h$ and $\mathsf{eval}_h$ can be computed in $\mathcal{O}(nt^2)$ and $\mathcal{O}(t^2)$ time, respectively. 
\end{theorem}
\begin{proof}
Since the $\sigma$-polynomial only needs to be stored modulo $z^{t+1}$, we can use the following algorithm: 

\RestyleAlgo{plain}
\LinesNumbered
\begin{algorithm}
\DontPrintSemicolon
Set $s \leftarrow 1$\;
\For{$i=1$ \KwTo $n$}
{
    $r \leftarrow (1 - a_i z)^{x_i} \pmod{z^{t+1}}$\;
    $s \leftarrow sr \pmod{z^{t+1}}$\;
}
\Return $s$\;
\end{algorithm}

Step 3 multiplies $1 - a_iz$ with itself up to $q$ times. Thus, this can be done in up to $q$ steps. Reduction modulo $z^{t+1}$ can be done via the division algorithm. Since this involves a polynomial of degree up to $q$ and another with degree $t+1$, this can be done in time $\mathcal{O}(qt)$~\cite[\S 17.1]{shoup-nt-book}. Step 4 involves multiplying two polynomials of degrees less than or equal to $t$. This can be done in $\mathcal{O}(t^2)$ time~\cite[\S 17.1]{shoup-nt-book}. Finally, reduction modulo $z^{t+1}$, as above, can be done in $\mathcal{O}(t^2)$ time, as $sr$ is of degree at most $2t$. Thus, $h$ can be computed in time $\mathcal{O}(n(qt + t^2 + t^2)) = \mathcal{O}(n t^2)$.

For the $\mathsf{eval}_h$ function, the first step is to find the inverse of $\sigma_{\vtr{x}}$ modulo $z^{t+1}$. This can be done using the EEA. The EEA takes time $\mathcal{O}(t^2)$~\cite[\S 17.3]{shoup-nt-book} as the polynomials are of degrees $t$ and $t+1$, respectively. This is followed by multiplication of degree $t$ polynomials $\sigma^{-1}_{\vtr{x}}$ and $\sigma_{\vtr{y}}$, and then by reduction modulo $z^{t+1}$, both taking $\mathcal{O}(t^2)$ time as discussed above. Finally, running the EEA algorithm on $\tilde{\sigma}_{\vtr{x}, \vtr{y}}$ and $z^{t+1}$ for at most $t_+ < t$ steps again takes time $\mathcal{O}(t^2)$. Thus, $\mathsf{eval}_h$ can be computed in $\mathcal{O}(t^2)$ overall time. 
\end{proof}

\subsection{Application to Adversarial Image Detection}
\label{subsec:application-adv}
For the adversarial image search scenario, we first need a setup algorithm to store the $\sigma$-polynomials, or rather their inverses, of all images $\vtr{x}_i \in \database$. At the time of submitting an input image, the user needs to prepare the $\sigma$-polynomial for his/her image to send to the server. Finally, the server runs the detection algorithm which evaluates to $1$ if for any image in the dataset we have a match according to the asymmetric $\ell_1$-distance predicate. These algorithms are detailed in Algorithms~\ref{algo:setup},~\ref{algo:prepare} and ~\ref{algo:detect}, respectively.  

\begin{algorithm}
%\SetAlgorithmName{Algorithm}{List of schemes}
\LinesNumbered
\SetAlgoLined
\DontPrintSemicolon
\SetKwInOut{Input}{Input}
\Input{All inputs to Construction~\ref{construct:pph}, database $\database$ of $N$ images $\vtr{x}_1, \ldots, \vtr{x}_{N}$.}
Run $\mathsf{samp}(1^\lambda)$ to obtain the hash function $h$\;
\For{$i=1$ \KwTo $N$}
{
    Obtain $h(\vtr{x}_i) = \sigma_{\vtr{x}_i}(z) \pmod{z^{t+1}}$\;
    Compute $\sigma^{-1}_{\vtr{x}_i}(z)$ modulo ${z^{t+1}}$\;
    Replace $\vtr{x}_i$ with $h(\vtr{x}_i)^{-1} = \sigma^{-1}_{\vtr{x}_i}(z)$ in $\database$\; 
}
\Return $\database$\; 
\caption{Setup}
\label{algo:setup}
\end{algorithm}

\begin{algorithm}
%\SetAlgorithmName{Algorithm}{List of schemes}
\LinesNumbered
\SetAlgoLined
\DontPrintSemicolon
\SetKwInOut{Input}{Input}
\Input{Hash function $h$ from Construction~\ref{construct:pph}, image $\vtr{y}$.}
\Return $h(\vtr{y}) = \sigma_{\vtr{y}}(z) \pmod{z^{t+1}}$\; 
\caption{Prepare}
\label{algo:prepare}
\end{algorithm}

\begin{algorithm}
%\SetAlgorithmName{Algorithm}{List of schemes}
\LinesNumbered
\SetAlgoLined
\DontPrintSemicolon
\SetKwInOut{Input}{Input}
\Input{Hash function $h$ from Construction~\ref{construct:pph}, hash digest $h(\vtr{y})$, database $\database$ of inverse $\sigma$-polynomials $\sigma^{-1}_{\vtr{x}_1} = h(\vtr{x}_1)^{-1}, \ldots, \sigma^{-1}_{\vtr{x}_N} = h(\vtr{x}_N)^{-1}$.}
\For{$i=1$ \KwTo $N$}
{
    \If{$\mathsf{eval}_h(h(\vtr{x}_i)^{-1}, h(\vtr{y})) =  1$}{
        \Return $1$\;
    }
}
\Return $0$\; 
\caption{Detect}
\label{algo:detect}
\end{algorithm}

\subsection{Information Leakage and Inverting the Hash Function}
\label{subsec:sigma-polynomials-reveal}

\descr{What Does the $\mathsf{eval}_h$ Function Reveal?} From Theorem~\ref{the:eea}, the EEA returns two polynomials $\alpha(z)$ and $\beta(z)$. Although not explicitly stated, these polynomials could be exactly the polynomials $\sigma_{\vtr{y} \dotdiv \vtr{x}}$ and $\sigma_{\vtr{x} \dotdiv \vtr{y}}$. We can then factor these polynomials using a variety of efficient polynomial factorization algorithms~\cite{shoup-nt-book}. From these factorized polynomials and knowing $\vtr{x}$ or $\vtr{y}$ one can recover the other, given that the vector $\vtr{a}$ is public information. This is obviously not surprising as this method was initially proposed to correct errors~\cite{tallinil1codes}. However, if the initial images $\vtr{x}$ and $\vtr{y}$ are not known, one only learns the absolute difference in pixel values between $\vtr{x}$ and $\vtr{y}$ for all pixels. This is certainly more information than a simple 1 or 0 answer to the fact that $\vtr{x}$ and $\vtr{y}$ satisfy $P_\mathsf{as}$. Since the server computes the $\mathsf{eval}_h$ function, we can reduce this information leakage by never storing the vector $\vtr{a}$ at the server. Otherwise, our scheme assumes an honest server. Note that we do not need to store the vector $\vtr{a}$ at the server, as the server can compute the $\mathsf{eval}_h$ function without it. Alternatively, if the honest assumption is too strong, we can employ threshold cryptography such that the $\sigma$-polynomial is secret shared over multiple servers, and the $\mathsf{eval}_h$ function is computed in a distributed manner. We leave out the exact workings of such a scheme as future work.

\descr{What Do the $\sigma$-Polynomials Reveal?}
The question then arises how difficult it is for an adversary to find an original image if it gets hold of the dataset $\database$? In Section~\ref{subsec:require} we mentioned that one of the adversarial goals is to invert the hash digest $h(\vtr{x})$. To be precise, for each image $\vtr{x} \in \database$, we store $\sigma_{\vtr{x}}^{-1}(z) \pmod{z^{t+1}}$, which is a $t$-degree polynomial. We can easily compute its inverse to obtain $\sigma_{\vtr{x}}(z) \pmod{z^{t+1}}$, which is again a degree $t$ polynomial. The general form of this polynomial is: 
\[
 \sigma_\vtr{x}(z) = A_{t}z^t + \cdots + A_1z + 1
\]
So, the question reduces to what do the coefficients $A_i$'s reveal about $\vtr{x}$? 
% However, if the vector $\vtr{a}$ is not known, and hence which element in the image corresponds to which element in $\vtr{a}$, then even though the multiplicity of the factors tell the program that reconstructs these polynomials that $\vtr{x}$ and $\vtr{y}$ differ by the given amount on some pixel. It is even possible to hide this information.\todo{Add the extended scheme}  
% Another thing we can do is to create shares of $\sigma$-polynomials instead of saving them in the clear. 
% In  In other words, given 
% How about the stored information about images $\vtr{x} \in \database$? 
% The question then is what do the coefficients of $\sigma(\vtr{x}) \pmod{z^{t+1}} = A_t z^t + \cdots + A_1 z + 1$, reveal about $\vtr{x}$? 
To answer this, we have the following theorem:
\begin{theorem}
\label{theo:relabel}
Let $\vtr{x}$ be an image with the $\sigma$-polynomial:
\[
\sigma_{\vtr{x}}(z) = \prod_{i=1}^n (1 - a_iz)^{x_i}.
\]
as given by Eq.~\eqref{eq:sigma-x}. Let $m = \sum_{i=1}^n x_i$ be the degree of this polynomial as given by Proposition~\ref{prop:sigma-x-facts}. Let $A_j$ be the $j$th coefficient of this polynomial, with $0 \leq j \leq m$. Then, 
\[
A_j = (-1)^j S(j, n),
\]
where 
\[
S(j, n - k) = \sum_{i = 0}^j \binom{x_{k+1}}{i}a_{k+1}^i S(i, n-k-1),
\]
for $0 \leq k \leq n$. 
\end{theorem}
\begin{proof}
Consider the polynomial in Eq.~\eqref{eq:sigma-x} for the vector $\vtr{x} = (x_1, \ldots x_n)$:
\[
\prod_{i=1}^n (1 - a_iz)^{x_i}.
\]
Let us relabel the $a_i$'s so that multiple occurrences of $a_i$'s have different labels, as shown in Figure~\ref{fig:relabel}. 

\begin{figure*}
\centering
\begin{tikzpicture}
\tkzDefPoint(0,0){a};
\node at (a)[]{$\vtr{a}$};

\node at (a)[below=1cm]{$\vtr{a}'$};

\tkzDefPoint(1,0){eqs};
\node at (eqs)[]{$=$};

\node at (eqs)[below=1.1cm]{$=$};

\tkzDefPoint(2,0){a1};
\node at (a1)[]{$(a_1,$};
\node at (a1)[below=0.4cm]{$\updownarrow$};
\node at (a1)[below=1cm]{$(a'_1,$};

\tkzDefPoint(3,0){dots1};
\node at (dots1)[]{$\cdots$};
\node at (dots1)[below=1.1cm]{$\cdots$};

\tkzDefPoint(4,0){a12};
\node at (a12)[]{$a_1,$};
\node at (a12)[below=0.4cm]{$\updownarrow$};
\node at (a12)[below=1cm]{$a'_{x_1},$};

\draw[gray,decoration={brace, raise=10pt},decorate](a1) -- node[above=10pt] {$x_1$}(a12);

\tkzDefPoint(5,0){a2};
\node at (a2)[]{$a_2,$};
\node at (a2)[below=0.4cm]{$\updownarrow$};
\node at (a2)[below=1cm]{$a'_{x_1 + 1},$};

\tkzDefPoint(6,0){dots2};
\node at (dots2)[]{$\cdots$};
\node at (dots2)[below=1.1cm]{$\cdots$};

\tkzDefPoint(7,0){a22};
\node at (a22)[]{$a_2,$};
\node at (a22)[below=0.4cm]{$\updownarrow$};
\node at (a22)[below=1cm]{$a'_{x_1 + x_2},$};

\draw[gray,decoration={brace, raise=10pt},decorate](a2) -- node[above=10pt] {$x_2$}(a22);

\tkzDefPoint(8,0){dots3};
\node at (dots3)[]{$\cdots$};
\node at (dots3)[below=1.1cm]{$\cdots$};

\tkzDefPoint(9,0){an};
\node at (an)[]{$a_n,$};
\node at (an)[below=0.4cm]{$\updownarrow$};
\node at (an)[below=1cm]{$a'_{m - x_n + 1},$};

\tkzDefPoint(10,0){dots4};
\node at (dots4)[]{$\cdots$};
\node at (dots4)[below=1.1cm]{$\cdots$};

\tkzDefPoint(11,0){an2};
\node at (an2)[]{$a_n)$};
\node at (an2)[below=0.4cm]{$\updownarrow$};
\node at (an2)[below=1cm]{$a'_{m})$};

\draw[gray,decoration={brace, raise=10pt},decorate](an) -- node[above=10pt] {$x_n$}(an2);
\end{tikzpicture}
    \caption{Replacing the labels in the tuples $\vtr{a}$ with unique labels.}
    \label{fig:relabel}
\end{figure*}

Here we have overloaded notation to also use $\vtr{a}$ to denote the tuple containing multiple occurrences of elements of the vector $\vtr{a}$. Then the polynomial can be rewritten in terms of $\vtr{a}'$ as:
\[
\prod_{i = 1}^m (1 - a'_iz),
\]
where $m = \sum_{i = 1}^n x_i$. Then the $j$th coefficient of this polynomial is given by the elementary symmetric polynomial~\cite[\S 8]{toecc-book}:
\begin{align*}
   e_j(a_1, a_2, \ldots, a_n) &= e_j(a'_1, a'_2, \ldots, a'_n) \\
   &= (-1)^j \sum_{1 \leq i_1 < i_2 < \cdots < i_j \leq n} a'_{i_1}a'_{i_2}\cdots a'_{i_j} 
\end{align*}
Let us call the sum on the right $S(j, n)$. We would like to get an expression of $S(j, n)$ in terms of the original vector $\vtr{a}$. Now $S(j, n)$ is the sum in the $j$th elementary symmetric polynomial in terms of the elements of the \emph{tuple} $\vtr{a}$, $S(j, n-1)$ is the sum in the $j$th elementary symmetric polynomial in terms of the elements of the \emph{tuple} $\vtr{a}$ \emph{without} $a_1$, and so on. Under this notation $S(0, i) = 1$ for all integers $i \geq 0$, and $S(i, i -1) = 0$ for all integers $i \geq 1$. 

Each summand in $S(j, n)$ is a product of $j$ elements of the tuple $\vtr{a}$. To calculate $S(j, n)$, first consider $a_1$. There are a total of $x_1$ occurrences of $a_1$ in the tuple $\vtr{a}$. Taken $j$ at a time, we therefore have a total of $\binom{x_1}{j}$ occurrences of $a_1^j$ in $S(j, n)$. Next we consider $a_1^{j-1}$ with the last coefficient being any of the other coefficients. We can have $\binom{x_1}{j - 1}$ possible arrangements that yield $a_1^{j-1}$. For each of these arrangements we need to determine the last coefficient in the $j$-term product. We are left with $n - 1$ coefficients: $a_2, \ldots, a_n$ and we are taking them one at a time. Thus, we are computing the quantity $S(1, n-1)$. Likewise for $a^{j-2}$ we need to consider the number of possible arrangements that yield $a^{j-2}$ which are $\binom{x_1}{j-2}$ and the number of possible ways in which the last two spots can be filled by the remaining $n-1$ elements, which is $S(2, n-1)$. Continuing on this way, once we reach $a_1^0$, we see that all $j$ spots in the product are taken by the rest of the elements in $\vtr{a}$. Thus, we are computing $S(j, n - 1)$. Collecting these counts, we get
\begin{align}
    S(j, n) &= \binom{x_1}{j}a_1^j + \binom{x_1}{j-1}a_1^{j-1}S(1, n - 1) + \binom{x_1}{j-2}a_1^{j-2} S(2, n - 1) \nonumber\\
    &+ \cdots + \binom{x_1}{1}a_1 S(j-1, n-1) + \binom{x_1}{0}a_1^0 S(j, n-1) \nonumber\\
    &= \sum_{i = 0}^j \binom{x_1}{i}a_1^i S(i, n-1) \nonumber.
\end{align}
From this equation we see that:
\[
S(j, n-1) = \sum_{i = 0}^j \binom{x_2}{i}a_2^i S(i, n-2),
\]
and so on. Thus, 
\[
A_j = e_j(a_1, \ldots, a_n) = (-1)^j S(j, n).
\]
\end{proof}

So, for example $S(0, n) = S(0, n-1) = \cdots = S(0, 0) = 1$, and hence $A_0 = 1$. Likewise, $S(1, n) = x_1 a_1 S(0, n-1) + S(1, n-1) = x_1a_1 + S(1, n - 1)$. By the recursive nature of the definition, we get $S(1, n -1) = x_2 a_2 + S(1, n -2)$. Continuing on, we get $S(1, 1) = x_na_n + S(1, 0) = x_n a_n$. Thus, $S(1, n) = \sum_{i=1}^n x_i a_i$, and so $A_1 = - \sum_{i=1}^n x_i a_i$. While these values can be easily computed if we know the vectors $\vtr{a}$ and $\vtr{x}$, not knowing the later means that we need to try $q^n$ possibilities, i.e., all possible values of $x_1$, times all possible values of $x_2$, and so on, to see which ones match $A_j$. Thus, the complexity of finding $\vtr{x}$ from the $\sigma$-polynomial of $\vtr{x}$ through this way is proportional to $\mathcal{O}(q^n)$. 

\descr{Are Collisions Possible?} As discussed in Section~\ref{subsec:correctness}, given the $\sigma$-polynomial $\sigma_{\vtr{x}}$ of an image $\vtr{x}$, it is trivial to find a collision in the \emph{polynomial space}. For instance one such collision is $\sigma_\vtr{x}$ and $z^{t+1} + \sigma_\vtr{x}$. For the same reason, it is also possible to launch a second preimage attack by computing $\sigma_\vtr{x}$ of any image $\vtr{x}$ and then computing $z^{t+1} + \sigma_\vtr{x}$, for instance. However, in both cases as discussed above, the resulting polynomial needs to be a valid $\sigma$-polynomial of some image for it to be a collision in the \emph{image space}. This is not straightforward as the brute-force way to find collisions in the image space is computationally expensive as the analysis after Theorem~\ref{theo:relabel} shows.

\descr{How Random is the Digest?} Theorem~\ref{theo:relabel} also sheds light on the format of the hash digest, i.e., the polynomial $\sigma_\vtr{x}$, of an image $\vtr{x}$. For instance, changing just one pixel of the image, say the $i$th pixel from $x_i$ to $x_i \pm 1$, where $1 \leq i \leq n$, changes each coefficient $A_j = (-1)^j S(j, n)$ of $\sigma_\vtr{x}$, as each term $S(j, n)$ is affected. In order to demonstrate this, we used the parameters $n = 28 \times 28 = 784$, $q = 256$ and $t = 2007$. With the prime $p = 787$, we sampled a random vector $\vtr{a}$, and further sampled 1,000 random images $\vtr{x}$. We then created four sets of 1,000 images as follows. The first set contained images $\vtr{x}'$ which had a single random pixel from $\vtr{x}$ changed by a value of 1; the second set changed the value of 10 random pixels in $\vtr{x}$ by 1; the third set changed the value of 100 random pixels by 1; and the fourth set contained random images drawn independently of $\vtr{x}$. Each image $\vtr{x}$ in the original set was compared with the corresponding images in the four sets by computing the Hamming distance between the coefficients of the resulting $\sigma$-polynomials. The set of images with a single pixel change had an average Hamming distance of $2004.379$, ten pixel changes resulted in an average Hamming distance of $2004.371$, hundred pixel changes resulted in an average Hamming distance of $2004.371$, and completely random images had an average Hamming distance of $2004.49$. Note that the degree of the polynomial is $2007$, but with one coefficient fixed to 1. This gives evidence that the digests are random, and even a single pixel change can completely change the digest. 

While this provides some evidence that the hash digest is random, we would like to emphasize that the digest itself leaks information: namely whether a given image is within a certain $\ell_1$-distance of a target image as this information is encoded in the digest. Thus, given a set of $\sigma$-polynomials of images, it is possible to classify whether a given target image is similar, in the sense of $\ell_1$-distance, to any image in the list. Thus, by definition, a property-preserving hash function does not prevent these types of classification attacks. 

%
% ---- Background and related works ----
%
\section{Experimental Results}
We use the Imagenette dataset~\cite{imagenette} (Imagenette-320 to be specific) for our adversarial image detection application, which is a smaller subset of the well-known Imagenet dataset~\cite{deng2009imagenet}. This is a dataset of 9,459 RGB images of size $224 \times 224$. Thus $n = 224 \times 224 \times 3$. 

\descr{Similarity Metrics.} To measure the difference between original and adversarial images we use three similarity metrics: (a) the Learned Perceptual Image Patch Similarity (LPIPS) metric~\cite{zhang2018perceptual}, which is widely used as a proxy for human perceptual similarity. Generally, a perturbed image is similar to the original one when its LPIPS is lower than 0.2, and the difference can be significantly perceived  when it is more than 0.3~\cite{zhang2018perceptual}, (b) Pixel Change Ratio, which shows the percentage change in absolute pixel values compared to the original image, and (c) normalized asymmetric $\ell_1$-distance (NAD) for two images $\vtr{x}$ and $\vtr{x}^*$ defined as
\begin{equation}
\label{eq:nad}
    \text{NAD}(\vtr{x}, \vtr{x}^*) = \frac{\max\{\norm{\vtr{x} \dotdiv \vtr{x}^*}_1, \norm{\vtr{x}^* \dotdiv \vtr{x}}_1\}}{qn} \times 100
\end{equation}
Recall that the maximum possible distance between two images is $(q-1)n \approx qn$. In this section we assume that $t_+ = t_- = \frac{t}{2}$. Note that for $P_\textsf{as}( \vtr{x}, \vtr{x}^*)$ to evaluate to 1, we must have:
\begin{equation}
\label{eq:nad-t-relation}
t_+ = t_- = \frac{t}{2} \leq \frac{qn \times \text{NAD}(\vtr{x}, \vtr{x}^*)}{100} \Rightarrow t \leq \frac{qn \times \text{NAD}(\vtr{x}, \vtr{x}^*)}{50}
\end{equation}
%Given a value of $t = \gamma qn$, for some $\gamma$, note that $\text{ND}(t) = \gamma \times 100$. Thus, we may measure $t$ in terms of the normalized distance ND. 
\subsection{\texorpdfstring{Possible Values of $t$}{Possible Values of t}}
\label{subsec:possible-t}
Ideally, the threshold $t$ for the asymmetric $\ell_1$-distance predicate $P_\mathsf{as}$ should be such that for any two images in the database $\database$ the predicate evaluates to 0, thus ensuring zero false positives. To do so, we define the \emph{empirical error} on the database $\database$ as: 
\begin{equation}
\label{eq:emp-error-pas}
\text{err}_{P_{\mathsf{as}}}(\database) = \sum_{1 \leq i < j \leq N} \frac{P_{\mathsf{as}}(\vtr{x}_i, \vtr{x}_j)}{\binom{N}{2}},
\end{equation}
and the parameters of $P_{\mathsf{as}}$ are $t_+ = t_- = \frac{t}{2}$ and $\delta = 3$. Although in general $P_{\mathsf{as}}(\vtr{x}_i, \vtr{x}_j) \neq P_{\mathsf{as}}(\vtr{x}_j, \vtr{x}_i)$, with $t_+ = t_-$ and large $t$, the difference is not profound enough to matter. We plot the error $\text{err}_{P_{\mathsf{as}}}(\database)$ in Figure~\ref{fig:err-pas} against increasing values of $t$. For all values of $t \leq 325,000 = 3.25 \times 10^5$ we get $\text{err}_{P_{\mathsf{as}}}(\database) = 0$. This is $t \approx 0.008qn$, or $\text{NAD} \approx 0.4$ from Eq.~\eqref{eq:nad-t-relation}. Thus, any value of $t$ less than this should produce no false positives for images in this dataset. Unfortunately, as we shall see, this value of $t$ is not high enough to defend against some attacks. We note that until $t = 2,000,000$ which is $t \approx 0.05qn$ or $\text{NAD} \approx 2.5$, the error rate is less than $2\%$. Thus, we can discard the few images that cause ``collisions'' which is most likely because these images are similar, to use a higher value of $t$.

\begin{figure}[!ht]
    \centering
    \includegraphics[scale=0.4]{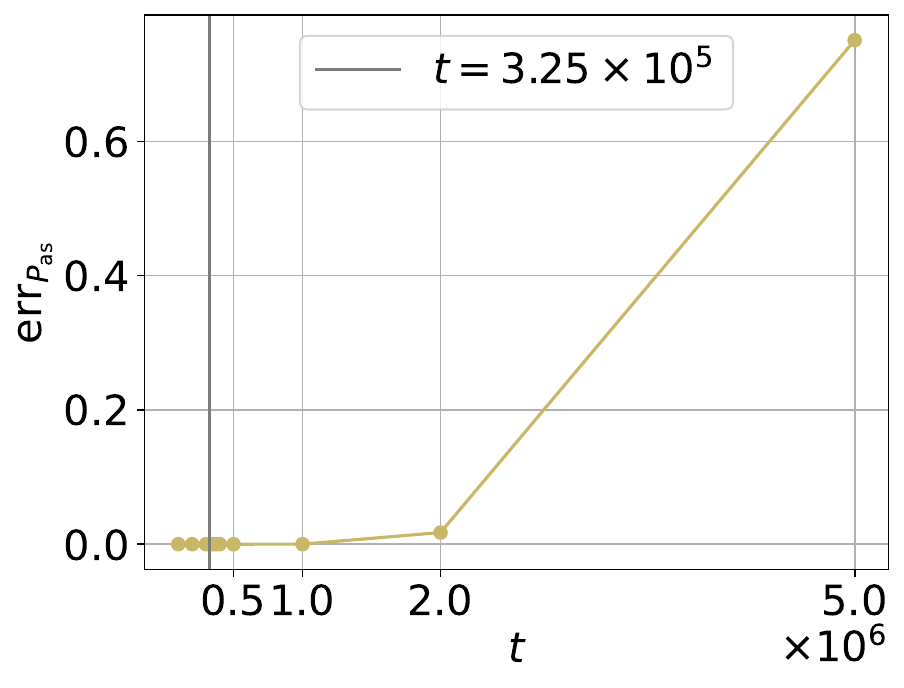}
    \caption{The empirical error on the Imagenette dataset $\database$. We get non-zero error with $t \geq 350,000 = 3.25 \times 10^5$.}
    \label{fig:err-pas}
\end{figure}

\subsection{Impact of Adversarial Attacks}
The Fast Gradient Sign Method (FGSM)~\cite{goodfellow2014explaining} and Projected Gradient Descent (PGD)~\cite{madry2018towards} are two well-known adversarial input attacks. In both attacks the noise parameter $\epsilon$ can be adjusted to add more noise to the input image, with the cross-entropy loss as the objective function. Figure~\ref{fig:fgsm-lpips} shows the impact on the quality of an image through LPIPS, pixel change ratio, and NAD as we increase $\epsilon$ in FGSM. At $\epsilon = 0.4$, the NAD is 0.4512, which approximately corresponds to $t = 325,000$ from Eq~\eqref{eq:nad-t-relation}. At this $t$, the LPIPS is 0.1786. At $\epsilon = 0.1$, with NAD = 1.1277, which corresponds to $t \approx 869,122$, we have LPIPS 0.5714, and hence the image quality has significantly degraded. 

%Due to lack of space, a similar graph for PGD is relegated to Figure~\ref{fig:all-lpips} in Appendix~\ref{app:res}.

\begin{figure*}[!ht]
    \centering
    % \begin{subfigure}[b]{0.45\textwidth}
    %     \centering
        \includegraphics[width=1.0\textwidth]{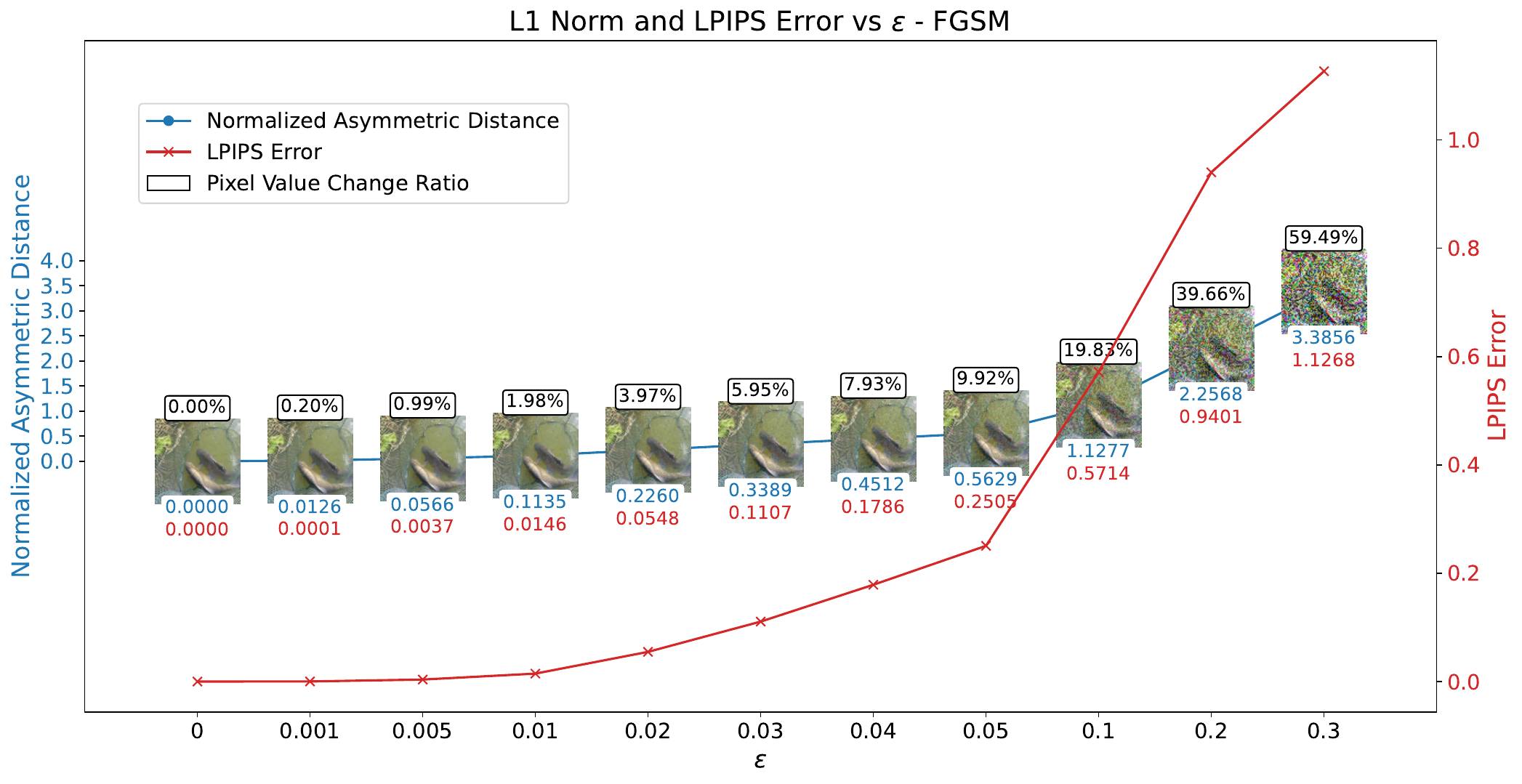}
        \label{fig:fgsm-lpips-1}
%    \end{subfigure}
    % \hfill
    % \begin{subfigure}[b]{0.45\textwidth}
    %     \centering
    %     \includegraphics[width=\textwidth]{figs/PGD_noise_0.pdf}
    %     \label{fig:fgsm-lpips-2}
    % \end{subfigure}
    % \vskip\baselineskip
    % \begin{subfigure}[b]{0.45\textwidth}
    %     \centering
    %     \includegraphics[width=\textwidth]{figs/Brightness_noise_0.pdf}
    %     \label{fig:fgsm-lpips-3}
    % \end{subfigure}
    % \hfill
    % \begin{subfigure}[b]{0.45\textwidth}
    %     \centering
    %     \includegraphics[width=\textwidth]{figs/Contrast_noise_0.pdf}
    %     \label{fig:fgsm-lpips-4}
    % \end{subfigure}
    \caption{The impact of adding noise to the image using the FGSM attack on the metrics LPIPS, pixel change ratio and NAD.}
    \label{fig:fgsm-lpips}
\end{figure*}

Figure~\ref{fig:all-lpips} shows the change in image quality as we continue to alter the original image using the PGD attack. To show the impact on a larger number of images, we sample $1,000$ images from the Imagenette dataset and apply both the FGSM and PGD attacks on them by varying $\epsilon$. The results are shown in Table~\ref{tab:fgsm_pgd_stats}. For FGSM, the average NAD of 1.1287, and for PGD an average NAD of 0.8163 produces LPIPS of 0.5596 and 0.4485, respectively, which is significant perceptual loss on the images. 

\begin{figure*}[!ht]
    \centering
    % \begin{subfigure}[b]{0.7\textwidth}
    %     \centering
    %     \includegraphics[width=\textwidth]{figs/FGSM_noise_0.pdf}
    %     \label{fig:fgsm-lpips-1}
    % \end{subfigure}
    % \hfill
    \begin{subfigure}[b]{0.85\textwidth}
        \centering
        \includegraphics[width=\textwidth]{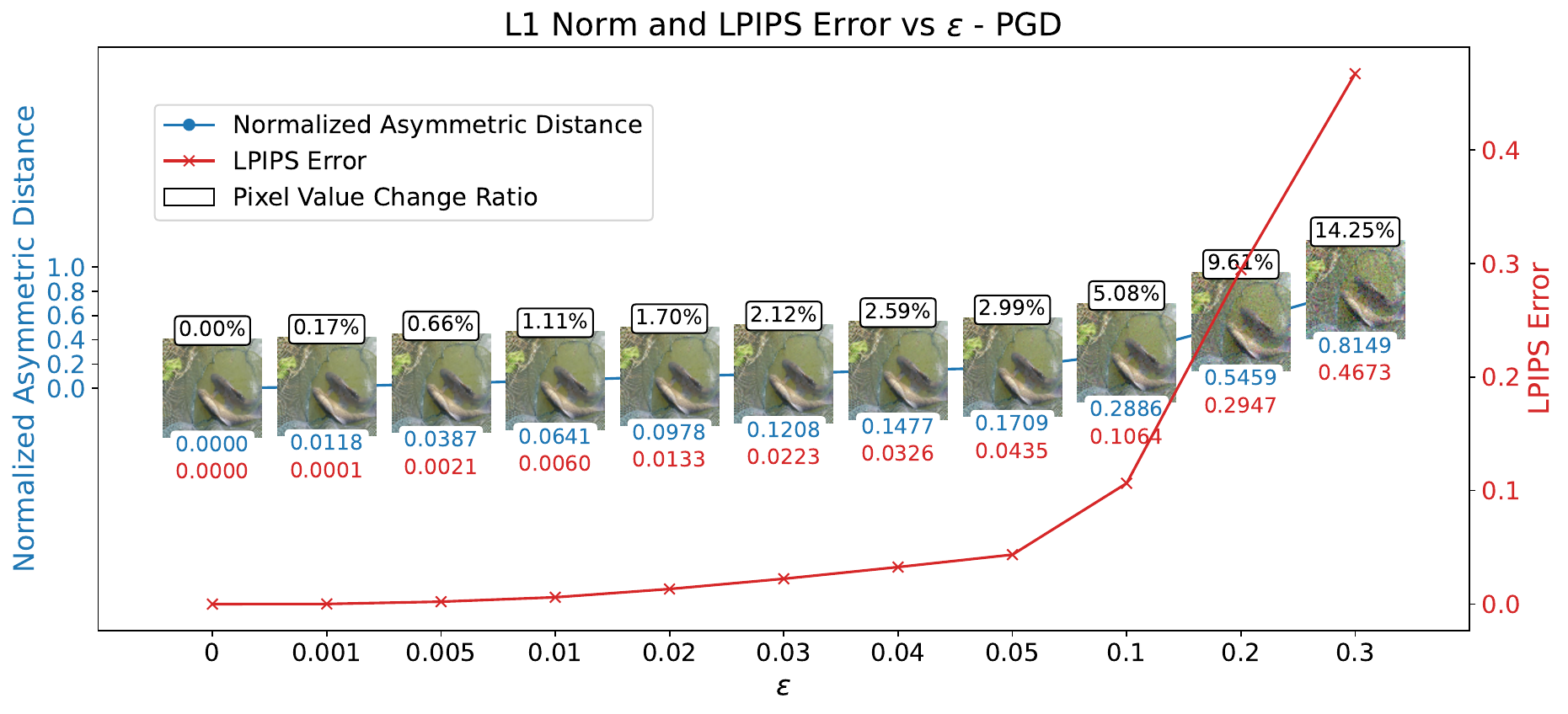}
        \label{fig:fgsm-lpips-2}
    \end{subfigure}
    \vskip\baselineskip
    \begin{subfigure}[b]{0.85\textwidth}
        \centering
        \includegraphics[width=\textwidth]{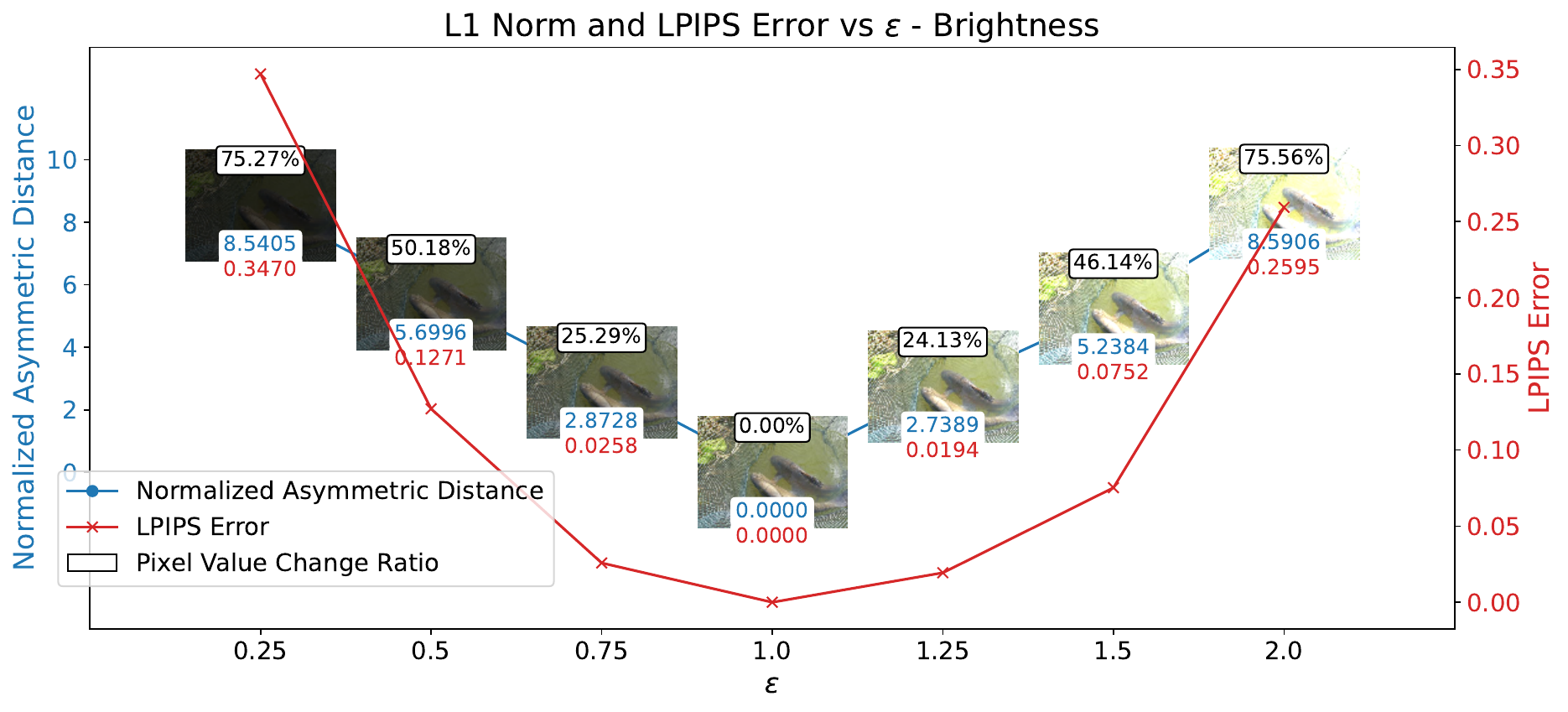}
        \label{fig:fgsm-lpips-3}
    \end{subfigure}
    \hfill
    \begin{subfigure}[b]{0.85\textwidth}
        \centering
        \includegraphics[width=\textwidth]{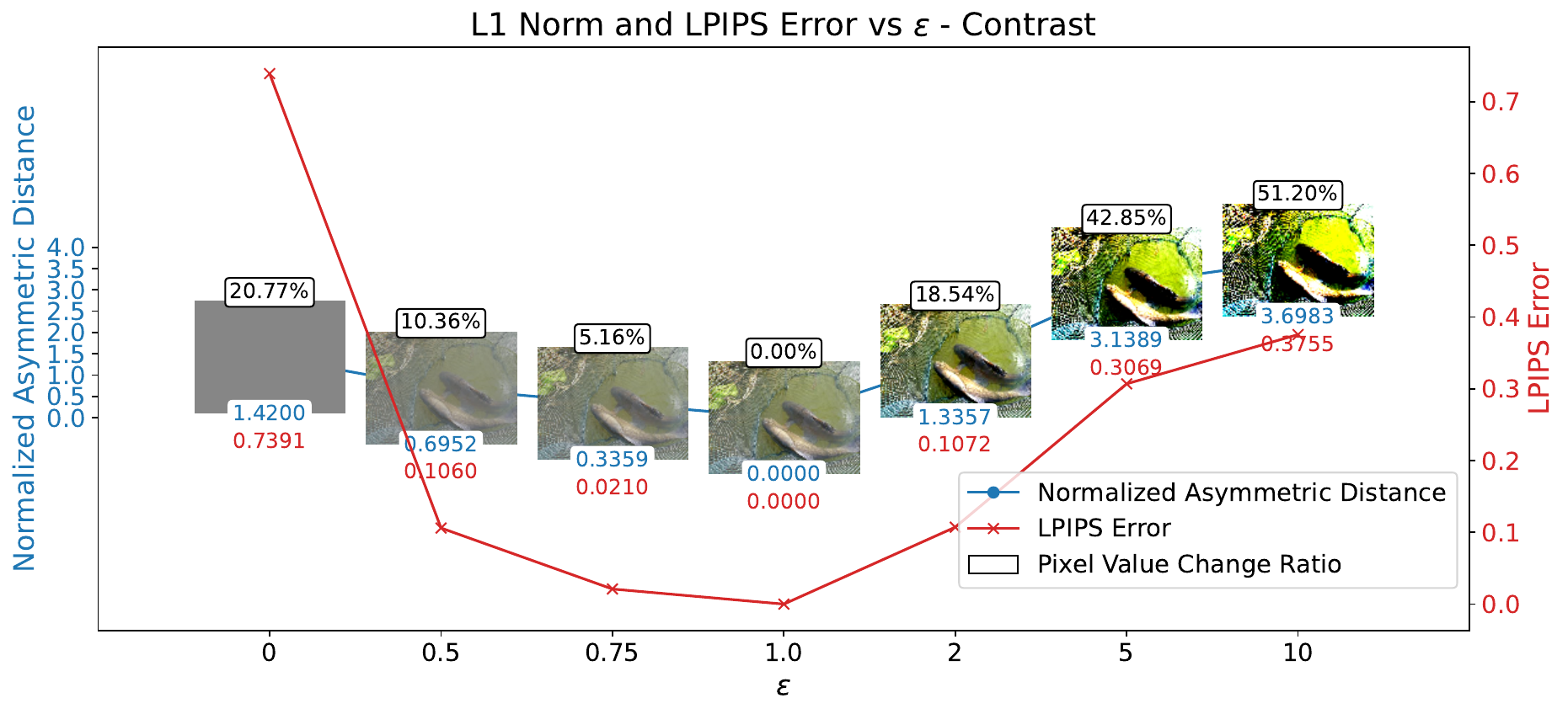}
        \label{fig:fgsm-lpips-4}
    \end{subfigure}
    \caption{The impact of increasing alterations to the image through the PGD attack, brightness, and contrast, on the metrics LPIPS, pixel change ratio and NAD.}
    \label{fig:all-lpips}
\end{figure*}

\begin{table}[h]
\centering
%\resizebox{\columnwidth}{!}{
\begin{tabular}{c|ccc|ccc}
$\epsilon$ & Attack & LPIPS & NAD & Attack & LPIPS & NAD \\
\hline \hline
0.000 & FGSM & 0.0000 & 0.0000 & PGD & 0.0000 & 0.0000 \\
0.005 &      & 0.0026 & 0.0580 &     & 0.0016 & 0.0379 \\
0.020 &     & 0.0524 & 0.2266 &     & 0.0145 &  0.1005 \\
0.050 &     & 0.2454 & 0.5644 &     & 0.0470 & 0.1768 \\
0.100 &     & 0.5596 & 1.1287 &     & 0.1187 & 0.2973 \\
0.200 &     & 0.9913 & 2.2576 &     & 0.2871 & 0.5523 \\
0.300 &     & 1.2105 & 3.3857 &     & 0.4485 & 0.8163
\end{tabular}
%}
\caption{Statistics for different $\epsilon$ values on FGSM and PGD.}
\label{tab:fgsm_pgd_stats}
\end{table}

\subsection{Impact of Image Transformations}
Apart from adversarial input attacks, an image can also be changed to evade detection via simple image transformations~\cite{struppek2022neuralhashbreak}. These include filters such as increasing brightness, adjusting contrast, rotating, and cropping the image. We select two such techniques: changing brightness and contrast. By adjusting an \emph{enhancement factor} we can vary the brightness and contrast, with a value of 1 giving the original image.\footnote{See \url{https://pillow.readthedocs.io/en/stable/reference/ImageEnhance.html}} We use the same symbol, i.e., $\epsilon$, to denote the enhancement factor. We show the impact of $\epsilon$ on the image quality using the same sample image in Figure~\ref{fig:all-lpips}. One thing to note is that through brightness and contrast adjustments, at certain $\epsilon$ values even though the quality of the image is both visibly and through the pixel change ratio metric, quite bad, the LPIPS is rather small. For instance, at $\epsilon = 0.5$, the LPIPS is only $0.1271$ for the brightness attack, even though the pixel change ratio is 50.18\%. Thus a smaller value of LPIPS may suffice for the brightness and contrast attacks to deteriorate the adversarial image. In Table~\ref{tab:brightness_contrast_stats} we show the average LPIPS and NAD on the 1,000 images chosen in our experiment. For both techniques, $\epsilon = 0.250$ results in LPIPS of more than $0.3$.

\begin{table}[h]
\centering
%\resizebox{\columnwidth}{!}{
\begin{tabular}{c|ccc|ccc}
$\epsilon$ & Attack & LPIPS & NAD & Attack & LPIPS & NAD \\
\hline \hline
0.250 & Brightness & 0.3183 & 6.8516  & Contrast & 0.3204 & 1.7338 \\
0.500 &      & 0.1140 & 4.5666 &     & 0.1109 & 1.1559 \\
0.750 &     & 0.0244 & 2.3057&     & 0.0236 & 0.5695  \\
1.000 &     & 0.0000 & 0.0000 &     & 0.0000 & 0.0000 \\
1.250 &     & 0.0258 & 1.9095 &     & 0.0179 & 0.5303 \\
1.500 &     & 0.0800 & 3.5170 &     & 0.0526 & 0.9300 \\
2.000 &     & 0.1955 & 5.8659 &     & 0.1208 & 1.4895
\end{tabular}
%}
\caption{Statistics for Brightness and Contrast attacks.}
\label{tab:brightness_contrast_stats}
\end{table}

\subsection{Computational Time}
We implement the PPH scheme using the Python library \verb+galois+~\cite{Hostetter_Galois_2020}. We set $t = 0.01qn$ which is more than the quantity $0.008qn$ established in Section~\ref{subsec:possible-t}. We choose $t_+ = t_- = t/2$, and $\delta = 3$. For small grayscale images, such as the $28 \times 28$ sized images used in the MNIST dataset~\cite{lecun1998mnist}, our scheme can compute the $\sigma$-polynomial of an image $\vtr{x}$ in about $0.26$ seconds. This is a one-off cost for the client, and hence not prohibitive. For database creation, we need to convert each image $\vtr{x} \in \database$ to its $\sigma$-polynomial and then invert it. This can be done in time $0.66$ seconds per image. Finally the $\mathsf{eval}_h$ function can be computed in $0.08$ seconds given the $\sigma$-polynomials of the input image and the inverse $\sigma$-polynomial of the target image. 

For larger image sizes, the time can grow large as our algorithms are of the order $\mathcal{O}(t^2)$. Our idea is to divide the image into blocks, and then compute the predicate $P_{\mathsf{as}}$ per block with $n_B = n/B$ being the size of each block and $t_B = t/B$ being the threshold per block, for a block size of $B$. Note that dividing the image into blocks is not unprecedented. It is done by the 
Blockhash perceptual hashing algorithm~\cite{yang2006block}, and this strategy is known to be more robust against local changes to images~\cite{schneider1996robustblock, hao2021pph-attack}. With this, for an RGB $224 \times 224$ image, the $\mathsf{eval}_h$ can be computed in about $0.013$ seconds per block or $13$ second per image. Detailed times are shown in Table~\ref{tab:times} which is the result of running the algorithms a total of 1,000 times. To calculate the time of $\mathsf{eval}_h$ we alter the image such that the predicate is not satisfied, which is the worst-case time. In Table~\ref{tab:eps-times}, we further show the time taken when we choose $t$ such that LPIPS is high for each of the four attacks. Here $t$ is chosen from NAD according to Eq.~\ref{eq:nad-t-relation}. The corresponding values of $\epsilon$ and LPIPS are taken from Tables~\ref{tab:fgsm_pgd_stats} and~\ref{tab:brightness_contrast_stats}.

\begin{table}[]
    \centering
\ifisfullpaper
    \resizebox{\columnwidth}{!}{
\fi
    \begin{tabular}{c|c|c|c|c|c|c|c}
    Size & Color & Blocks & $n_B$ & $t_B$ & Time $\sigma$ & Time $\sigma^{-1}$ & Time $\mathsf{eval}_h$\\
    \hline\hline

    $224 \times 224$	& RGB & 1000 &	150	& 384 & $0.0235$ &	$0.0963$ & $0.0128$ \\
    $128 \times 128$ & RGB	& 100 & 491	& 1256	& $0.1238$ & $0.3712$ &	$0.0455$ \\
    $64	\times 64$ & RGB &	100	& 122	& 312	& $0.0185$ & $0.0777$ &	$0.0101$ \\
    $28 \times 28$ & RGB & 10 & 235	& 601 &	$0.0436$ &	$0.1584$ & $0.0204$ \\
    $28	\times 28$ & Gray &	1 &	784	& 2007 & $0.2596$ &	$0.6667$ & $0.0784$    
    \end{tabular}
\ifisfullpaper
}
\fi
\caption{Total time in seconds taken by our scheme to produce the $\sigma$-polynomial of an image, the $\sigma$-polynomial and its inverse of an image, and the solution using the $\mathsf{eval}_h$ function, where the images are divided into blocks with $t = 0.01qn$.}
\label{tab:times}
\end{table}

\begin{table}
    \centering
\ifisfullpaper
    \resizebox{\columnwidth}{!}{
\fi
    \begin{tabular}{c|c|c|c|c|c|c|c|c}
    Attack & $\epsilon$ & LPIPS & NAD & $n_B$ & $t_B$ & Time $\sigma$ & Time $\sigma^{-1}$ & Time $\mathsf{eval}_h$\\
    \hline\hline
    FGSM &	 0.100 & 0.5596 &  1.1287 & 150 & 869 & 0.0303 &	0.1993 &	0.0309 \\ 
    PGD & 0.300 & 0.4485 & 0.8163 & 150	& 629	&	0.0280 & 0.1506 &	0.0220 \\
    Brightness	& 0.250 & 0.3183 &	6.8516 & 150	& 5280 &	0.0847 & 1.3360 &	0.2651 \\
    & 2.000 & 0.1955 & 5.8659	& 150 & 4520 & 0.0763 &	1.0865 &	0.2143 \\
    Contrast & 0.250 & 0.3204 & 1.7338	 & 150	& 1336	& 0.0386 & 0.2999 &0.0489 \\
    & 2.000 & 0.1208 & 1.4895 & 150 &	1147 &	0.0358 & 0.2592 &	0.0422
    \end{tabular}
\ifisfullpaper
}
\fi
\caption{Total time in seconds taken by our scheme to produce the $\sigma$-polynomial of an image, the $\sigma$-polynomial and its inverse of an image, and the solution using the $\mathsf{eval}_h$ function, where $t$ is chosen such that LPIPS is high for each attack.}
\label{tab:eps-times}
\end{table}

Our implementation was done on a standard Apple M3 ARM processor with 8 (performance) cores, and 16 GB RAM. We note that with more cores, and involving GPUs, these times can be substantially improved, as the algorithms are parallelizable: each block and each database image can be evaluated separately. 

\subsection{Summary of Parameters}
\label{subsec:params}
Due to a large number of parameters used in this paper, we have provided a summary of them and their role in Table~\ref{tab:params} to help the reader better navigate the scheme.

\begin{table}[h]
\centering
%\resizebox{\columnwidth}{!}{
\begin{tabularx}{\textwidth}{c|X}
Parameter & Description and Role \\
\hline \hline
$n$ & Size of an image, e.g., $n = 28 \times 28$; equals the number of pixels \\
$q$ & Range of pixel values, e.g., $q = 256$ for grayscale images \\
$t$ & Distance threshold; images within distance $t$ are considered similar\\
$m$ & Size of the hash digest\\
$p$ & Prime number greater than $n$ to construct the finite field $\mathbb{Z}_p$\\
$\vtr{a}$ & A vector of random, unique non-zero coefficients from $\mathbb{Z}_p$, one for each of the $n$ pixels\\
$\sigma_\vtr{x}$ & The $\sigma$-polynomial associated with an image $\vtr{x}$; this encodes the image as a polynomial with coefficients from $\mathbb{Z}_p$\\
$z^{t+1}$ & The monic polynomial which reduces the inverse of $\sigma_\vtr{x}$ to a degree $t$ polynomial; the reduced polynomial is stored as the encoded image \\
$\mathsf{eval}_h$ & The evaluation function which uses the extended Euclidean algorithm to determine whether the $\sigma$-polynomial of an input image is within asymmetric $\ell_1$-distance of $t$ from the target image \\
$t_+$ and $t_-$ & The two one-sided distance thresholds for the asymmetric $\ell_1$-distance predicate, interpreted as the set differences $\vtr{y} - \vtr{x}$ and $\vtr{x} - \vtr{y}$, respectively, if images $\vtr{x}$ and $\vtr{y}$ are represented as multisets \\
$\delta$ & A parameter to reduce errors if two images are dissimilar according to the asymmetric $\ell_1$-distance \\
$B$ & Number of blocks to divide larger images into smaller chunks for computational efficiency; $n$ and $t$ are updated per block as $n_B = n/B$ and $t_B = t/B$
\end{tabularx}
%}
\caption{Summary of Main Parameters.}
\label{tab:params}
\end{table}

\section{Related Work}
\label{sec:rw}
The notion of robust property-preserving hash (RPPH) functions was formally introduced by Boyle et al~\cite{boyle2018adversarially}, where they also give a construction of an RPPH for gap-Hamming distance predicate: the predicate outputs 1 if the Hamming distance is lower than one threshold, 0 if it is higher than the other, and a special symbol if it lies within the gap. The authors also show a construction of a non-rubust PPH for gap-Hamming distance predicate using a locality sensitive hashing (LSH) scheme from~\cite{kushilevitz1998efficient}. Following their work, new constructions for gap-Hamming distance as well as exact Hamming distance predicates have been proposed~\cite{holmgren2022nearly, fleischhacker2021robust, fleischhacker2022property}. In particular, the construction from~\cite{holmgren2022nearly} is based on the idea of efficient list decoding of linear codes. This inspired us to search for list decoding of errors measured in the Euclidean distance ($\ell_2$ distance). Mook and Peikert~\cite{mook2021lattice} show a construction of list decoding of error-correcting codes based on Reed-Solomon codes for $\ell_2$ distance. Unfortunately, as we show in Section~\ref{subsec:feasibility}, this procedure is unlikely to be efficient even for small values of the distance threshold $t$, as the size of the list (possible vectors) blows up. While our result is for the $\ell_1$ metric, from Proposition~\ref{prop:l1-l2-relation}, it also applies to the $\ell_2$ metric. Our construction is based on $\ell_1$-error correcting codes from~\cite{tallinil1codes}, which themselves are derived from their earlier construction of error-correcting codes for Hamming distance~\cite{tallinihammingcodes}. Several adversarial attacks have been shown against perceptual hashing algorithms~\cite{struppek2022neuralhashbreak, hao2021pph-attack}. In particular, evasion attacks add small adversarial noise to cause a mismatch in the hashes of two perceptually similar images. This attack does not apply to property-preserving hashing since the probability of a mismatch between the predicate on the original domain and the hash domain is required to be negligible.   

\section{Conclusion}
\label{sec:conclude}
We have proposed the first property-preserving hash (PPH) function family for (asymmetric) $\ell_1$-distance predicates, with applications to countering adversarial input attacks. Our construction is efficient, as shown through our implementation. Our work leaves a number of avenues for future research. While the proposed hash function shows high compression for smaller distance thresholds $t$, our theoretical results show that further compression may be possible, especially when $t$ is large. Furthermore, our scheme is only robust against one-sided errors and only handles asymmetric $\ell_1$-distance predicates. It remains an open problem to find a robust PPH for the exact $\ell_1$-distance predicate. There may also be interest in finding a robust PPH scheme for the Euclidean distance predicate. Additionally, it is also desirable to find a scheme that is provably hiding, i.e., which shows that inverting the hash function is computationally infeasible. Finally, the implementation of our scheme has the potential to be further sped up due to its highly parallel nature. 

\section*{Acknowledgments}
We are grateful to Quang Cao, Ron Steinfeld and Josef Pieprzyk for helpful discussions on earlier ideas on this topic.

\appendix
\section{Some Useful Results}
\label{app:useful}

\descr{Metrics.} Let $S$ be a set. A function $d: S \times S \rightarrow \mathbb{R}$ is called a metric on $S$ if for all $x, y, z \in S$, (1) $d (x, y) \ge 0$ with equality if and only if $x = y$, (2) $d(x, y) = d(y, x)$, and (3) $d(x, y) \leq d(x, z) + d(z, y)$~\cite{searcoid-metric}. In this case $S$ is called a metric space. 
% However, addition of two points in $\mathbb{Z}_q$ is not defined modulo $q$ in the usual sense. Instead we define it as:
% \begin{equation}
% \label{eq:image-add}
% x +_q y = \min\{x + y, q - 1\}
% \end{equation}
% where the addition on the right hand side is over reals. We define $\mathbb{Z}_q^n$ element-wise in an analogous manner. 
For $x, y \in \mathbb{Z}_q$, i.e., pixels, define the function $d: \mathbb{Z}_q \times \mathbb{Z}_q \rightarrow \mathbb{R}$ as: 
\begin{equation}
\label{eq:d}
d(x, y) = |x - y |
\end{equation}
It follows that $d$ is a metric on $\mathbb{Z}_q$, as can be easily verified.  
% The fact that $d(x, y) \geq 0$ is obvious. And since $x, y \in \mathbb{Z}_q$, equality is only possible if $x = y$. The identity $d(x, y) = d(y, x)$ is trivial. To see that $d(x, y) \leq d(x, z) + d(z, y)$, for all $x, y, z \in \mathbb{Z}_q$, note that $0 \leq |x - y | < q$ since both $x$ and $y$ are in $\mathbb{Z}_q$. The same goes for $|x - z|$ and $|z - y|$. Then from the triangle inequality for real numbers: 
% \begin{align*}
%     |x - y | &\leq |x - z | + |z - y | \\
% \Rightarrow d(x, y) &\leq d(x, z) + d(z, y).
% \end{align*}
% Note that the addition here is over reals. Thus, $\mathbb{Z}_q$ together with the metric $d$ is a metric space. 
Similarly, define the Hamming distance $d_H: \mathbb{Z}_q \times \mathbb{Z}_q \rightarrow \mathbb{R}$ as: 
\begin{equation}
    d_H(x, y) = \begin{cases}
        1, \text{ if } x \neq y, \\
        0, \text{ otherwise}
    \end{cases}
\end{equation}
Then $d_H$ is also a metric. 
%It is easy to check that $d_H(x, y) \geq 0$ with equality only if $x = y$. It is also trivial to see that $d_H(x, y) = d_H(y, x)$. To prove that $d_H(x, y) \leq d_H(y, z) + d_H(z, y)$, we check all cases. If $x = y$, then $d_H(x, y) = 0$, and the right hand side will be equal or larger due to the first property. So assume that $x \neq y$. Then $d_H(x, y) = 1$. We can have $x = z$ or $z = y$, but not both as this would mean $x  = y$. Thus the right hand side is at least 1. On the other hand if $x \neq z$ and $z \neq y$, then the right hand side is 2. Proving the inequality. 
Thus, $\mathbb{Z}_q$ together with $d$ or with $d_H$ is a metric space. It follows that the following are metrics spaces on the set $\mathbb{Z}_q^n$~\cite[\S 1.6]{searcoid-metric}:
\begin{enumerate}
    \item The set $\mathbb{Z}_q^n$ together with the metric $d_0(\vtr{x}, \vtr{y}) = \sum_{i=1}^n d_H(x_i, y_i)$.
    \item The set $\mathbb{Z}_q^n$ together with the metric $d_1(\vtr{x}, \vtr{y}) = \sum_{i=1}^n d(x_i, y_i)$.
    \item The set $\mathbb{Z}_q^n$ together with the metric $d_2(\vtr{x}, \vtr{y}) = \sqrt{\sum_{i=1}^n (d(x_i, y_i))^2}$.
    \item The set $\mathbb{Z}_q^n$ together with the metric $d_\infty(\vtr{x}, \vtr{y}) = \max_{i} \{d(x_i, y_i)\}$,
\end{enumerate}
for all $\mathbf{x}, \mathbf{y} \in \mathbb{Z}_q^n$, with $x_i$ and $y_i$ being their $i$th elements. With these metrics we define the following norms for all $\mathbf{x} \in \mathbb{Z}_q^n$:

%That is, it is a vector of $n$ pixels. We define the following norms on the pixel subspace. Given $\mathsf{p} = (p_1, \ldots ,p_m)$:
\begin{enumerate}
    \item $\ell_0$-norm (Hamming weight): $\lVert \vtr{x} \rVert_0 = d_0(\mathbf{x}, 0)$.
    \item $\ell_1$-norm: $\lVert \vtr{x} \rVert_1 = d_1(\vtr{x}, \vtr{0})$.
    \item $\ell_2$-norm (Euclidean): $\lVert \vtr{x} \rVert_2 = d_2(\vtr{x}, \vtr{0})$.
    \item $\ell_\infty$-norm: $\lVert \vtr{x} \rVert_\infty = d_\infty(\vtr{x}, \vtr{0})$.
\end{enumerate}
%A generic norm will simply be denoted by $\norm{\cdot}$. 
Technically these are not norms as we have not even defined the vector space of images. However, we will continue to use the term by assuming the operations to be over the vector space $\mathbb{R}^n$.  

\descr{Miscellaneous Results.}
\begin{proposition}
\label{prop:aq-result}
Let $a > 1.4$ be a real number, and let $q \geq 4$ be a positive integer. Then $a^q - 1 \geq a^{q-1}$.
\end{proposition}
\begin{proof}
We prove this via induction. For $q = 4$, the left hand side is $a^4 - 1$ and the right hand side is $a^3$. For the inequality to hold we must have $a^3(a-1) \geq 0$. Since this is an increasing function of $a$, and direct substitution shows that this inequality is satisfied by $a = 1.4$, it follows that it is satisfied by all $a > 1.4$. 
Now suppose the result holds for $q = k$, i.e., $a^k - 1 \geq a^{k-1}$. Multiplying both sides by $a$, we get
\begin{align*}
    (a^k - 1) a \geq a^{k-1} a\\
    a^{k + 1} - a \geq a^{k} \\
    a^{k + 1} - 1 \geq a^{k} + a - 1 \geq a^k,
\end{align*}
where the last step follows from the fact that $a > 1$.
\end{proof}

\descr{Bernoulli's Inequality.} Let $x > -1$ be a real number and let $q$ be a positive integer then 
\begin{equation}
\label{eq:bernoulli-inequality}
(1 + x)^q \geq 1 + qx    
\end{equation}
See for example~\cite{bernoulli-inequality}. 
\begin{proposition}
\label{prop:increasing-function-q}
Let $q \geq 2$ be an integer. Then the function
\[
f(q) = \left( 1 + \frac{0.9}{q}\right)^q
\]
is an increasing function of $q$. 
\end{proposition}
\begin{proof}
Consider the ratio:\footnote{Part of this proof is derived from: \url{https://math.stackexchange.com/questions/1589429/how-to-prove-that-logxx-when-x1/}.}
\begin{align*}
    \frac{f(q + 1)}{f(q)} &= \frac{\left( 1 + \frac{0.9}{q+1}\right)^{q+1}}{\left( 1 + \frac{0.9}{q}\right)^q} = \frac{\left( 1 + \frac{0.9}{q+1}\right)^{q+1}}{\left( 1 + \frac{0.9}{q}\right)^q} \frac{\left( 1 + \frac{0.9}{q}\right)}{\left( 1 + \frac{0.9}{q}\right)} \\
    &= \left( \frac{q + 1 + 0.9}{q+1} \frac{q}{q + 0.9} \right)^{q+1} \left( 1 + \frac{0.9}{q}\right)\\
    &= \left( \frac{(q + 1 + 0.9)q + 0.9 - 0.9}{(q+1)(q+0.9)}\right)^{q+1} \left( 1 + \frac{0.9}{q}\right)\\
    &= \left( 1 + \frac{- 0.9}{(q+1)(q+0.9)}\right)^{q+1} \left( 1 + \frac{0.9}{q}\right)
\end{align*}
Now since $q \geq 2$, we have 
\[
\frac{- 0.9}{(q+1)(q+0.9)} \geq - \frac{0.9}{3 \times 2.9} > -1.
\]
Hence we can apply Bernoulli's inequality (Eq.~\ref{eq:bernoulli-inequality}) to obtain:
\begin{align*}
    \frac{f(q + 1)}{f(q)} &\geq \left( 1 - \frac{0.9}{q+0.9}\right) \left( 1 + \frac{0.9}{q}\right) = \left( \frac{q + 0.9 - 0.9}{q+0.9}\right) \left( \frac{q + 0.9}{q}\right) = 1.
\end{align*}
Thus $f(q+1) \geq f(q)$. 
\end{proof}

% \section{A Brief on Adversarial Attacks}
% \label{app:adv}

% \descr{FGSM.} The Fast Gradient Sign Method (FGSM)~\cite{goodfellow2014explaining}

% \descr{PGD.} The Projected
% Gradient Descent (PGD)~\cite{madry2018towards}

\bibliography{bibliography}

\end{document}